\documentclass[journal,12pt,english,draftcls,onecolumn]{IEEEtran}

\pdfoutput=1

\usepackage[T1]{fontenc}
\usepackage[latin9]{inputenc}
\usepackage{geometry}
\usepackage{color}
\usepackage{array}
\usepackage{float}
\usepackage{multirow}
\usepackage{amsthm}
\usepackage[cmex10]{amsmath}
\usepackage{amssymb}
\usepackage{graphicx}
\usepackage{esint}
\usepackage{geometry}
\usepackage{cite}

\makeatletter

\providecommand{\tabularnewline}{\\}

  \theoremstyle{plain}
  \newtheorem{prop}{\protect\propositionname}
  \theoremstyle{plain}
  \newtheorem{cor}{\protect\corollaryname}

\makeatother

\usepackage{babel}
  \providecommand{\propositionname}{Proposition}
\providecommand{\corollaryname}{Corollary}

\begin{document}

\title{Robust Bayesian compressed sensing\\
over finite fields: asymptotic performance analysis}

\author{Wenjie Li, Francesca Bassi, and Michel Kieffer
\thanks{W. Li and M. Kieffer are with Laboratoire des Signaux et Syst\`emes,
CNRS--Supelec--Univ Paris-Sud, Gif-sur-Yvette. M.~Kieffer is partly supported by the Institut Universitaire de France.
F. Bassi is with ESME-Sudria, Ivry-sur-Seine.
This work was partly supported by European Network of Excellence project
NEWCOM\#.}}

\maketitle

\begin{abstract}
This paper addresses the topic of robust Bayesian compressed sensing
over finite fields. For stationary and ergodic sources, it provides asymptotic 
(with the size of the vector to estimate) necessary and sufficient conditions on 
the number of required measurements to achieve vanishing reconstruction error, in 
presence of sensing and communication noise. In all considered cases, the necessary and 
sufficient conditions asymptotically coincide. Conditions on the sparsity of the 
sensing matrix are established in presence of communication noise. Several previously 
published results are generalized and extended.
\end{abstract}


\section{Introduction}

Compressed sensing refers to the compression of a vector $\boldsymbol{\theta}\in\mathbb{R}^{N}$,
obtained by acquiring linear measurements whose number $M$ can be
significantly smaller than the size of the vector itself. If $\boldsymbol{\theta}$
is $k$-sparse with respect to some known basis, its almost surely
exact reconstruction can be evaluated from the linear measurements
using basis pursuit, for $M$ as small as $\mathcal{O}(k\log(N))$
\cite{candes06,tao}. The same result holds true also for compressible
vector $\boldsymbol{\theta}$ \cite{compressedS}, with reconstruction
quality matching the one allowed by direct observation of the biggest
$k$ coefficients of $\boldsymbol{\theta}$ in the transform domain.
The major feature of compressed sensing is that the linear coefficients
do not need to be adaptive with respect to the signal to be acquired,
but can actually be random, provided that appropriate conditions on
the global measurement matrix are satisfied \cite{tao,candes05}.
Moreover, compressed sensing is robust to the presence of noise in
the measurements \cite{candes05,haupt06}.

Bayesian compressed sensing \cite{ji08} refers to the same problem,
considered in the statistical inference perspective. In particular,
the vector to be compressed is now understood
as a statistical source $\boldsymbol{\mathsf\Theta}$, whose \emph{a priori} distribution can induce
sparsity or correlation between the symbols. This allows to redefine
the reconstruction problem as an estimation problem, solvable using
standard Bayesian techniques, \emph{e.g.}, Maximum \emph{A Posteriori}
(MAP)
estimation. In practical implementations, estimation from the
linear measurements can be achieved exploiting statistical graphical
models \cite{montanari12}, \emph{e.g.}, using belief propagation
\cite{BP} as done in \cite{baron10} for deterministic measurement
matrices, and in \cite{bayati11} for random measurement matrices.

In this paper we address the topic of robust Bayesian compressed sensing
over finite fields. The motivating example for considering this setting
comes from the large and growing bulk of works devoted to data 
{dissemination
and collection} in wireless sensor networks. Wireless sensor networks
\cite{WSN} are composed by autonomous nodes, with sensing capability
of some physical phenomenon (\emph{e.g.} temperature, or pressure).
In order to ensure ease of deployment and robustness, the communication
between the nodes might need to be performed in absence of designated
access points and of a hierarchical structure. At the network layer,
dissemination of the measurements to all the nodes can be achieved
using an asynchronous protocol based on random linear network coding (RLNC) 
\cite{RLNC}. In the protocol, each node in the network broadcasts
a packet evaluated as the linear combination of the local measurement,
and of the packets received from neighboring nodes. The linear coefficients
are randomly chosen, and are sent in each packet header. Upon an appropriate
number of communication rounds, each node has collected enough linearly
independent combinations of the network measurements, and can perform
decoding, by solving a system of linear equations. Due to the physical
nature of the sensed phenomenon, and to the spatial distribution of
the nodes in the network, correlation between the measurements at
different nodes can be assumed, and exploited 
to perform decoding, as done in \cite{LL,EUSI,giannakis,bourtsoulatze12,nabaee12}.
Recasting the problem in the Bayesian compressed sensing framework,
the vector of the measurements at the nodes
is interpreted as the compressible source $\boldsymbol{\mathsf\Theta}$, the network coding matrix
as the sensing matrix, and the decoding at each node as the estimation
operation. 

Before transmission, all the measurements needs to be quantized.
Quantization can be performed after the network encoding operation, 
as done in \cite{nabaee12}, where reconstruction
on the real field is performed via $\ell_{1}$-norm minimization,
or it can be done prior to the network encoding operation. For the latter
choice, which is the target of this work, each quantization index
is represented by an element of a finite field, from which the network
coding coefficients (\emph{i.e.}, the sensing coefficients in the
compressed sensing framework) are chosen as well. This setting has been considered
in \cite{LL}, where exact MAP reconstruction is obtained solving
a mixed-integer quadratic program, and in \cite{EUSI,bourtsoulatze12,giannakis},
where approximate MAP estimation is obtained using variants of the
belief propagation algorithm.

The performance analysis of compressed sensing over finite fields
has been addressed in \cite{CS,giannakis}, and \cite{seong}. The work in \cite{seong}
does not consider Bayesian priors, and assumes a known sparsity level
of $\boldsymbol{\theta}$. Ideal decoding via $\ell_{0}$-norm
minimization is assumed, and necessary and sufficient conditions for exact recovery
are derived as functions of the size of the vector, its sparsity level,
the number of measurements, and the sparsity of the sensing matrix.
Numerical results show that the necessary and sufficient conditions
coincide, as the size of $\boldsymbol{\theta}$ asymptotically increases.
A Bayesian setting is considered in \cite{CS} and \cite{giannakis}.
In \cite{CS} a prior distribution
induces sparsity on the realization of $\boldsymbol{\mathsf\Theta}$,
whose elements are assumed statistically independent. Using the method
of types \cite{types}, the error exponent with respect to exact reconstruction
using $\ell_{0}$-norm minimization is derived in absence of noise
in the measurements, and the error exponent with respect to exact
reconstruction using minimum-empirical entropy decoding is derived for noisy
measurements. In \cite{giannakis} specific correlation patterns (pairwise,
cluster) between the elements of $\boldsymbol{\mathsf\Theta}$ are considered.
Error exponents under MAP decoding are derived, only in case of absence of noise on
the measurements.

The contribution of this work can be summarized as follows. We assume
a Bayesian setting and we consider MAP decoding. Inspired by the work
in \cite{seong}, we aim to derive necessary and sufficient conditions
for almost surely exact recovery of $\boldsymbol{\theta}$, as its
size asymptotically increases. We consider three classes of prior
distributions on the source vector: \emph{i}) the prior distribution
is sparsity inducing, and the elements are statistically independent;
\emph{ii}) the vector $\boldsymbol{\mathsf\Theta}$ is a Markov process;
\emph{iii}) the vector $\boldsymbol{\mathsf\Theta}$ is an ergodic process.
{To the best of our knowledge, no analysis has been previously performed
for the latter source model, which is quite general.}
We consider both sparse and dense
sensing matrices. We consider two kinds of noises: \emph{a}) the sensing
noise, affecting the measurements prior to network coding (\emph{i.e.,}
prior to random projection acquisition in the compressed sensing framework);
\emph{b}) the communication noise, affecting the network coded packets
(\emph{i.e.}, the random projections in the compressed sensing framework).
{To the best of our knowledge, no analysis has been previously performed
in presence of both kinds of noise.}
Considering source model \emph{i}), our results for the noiseless
setting are compatible with the ones presented in \cite{seong}; in addition,
we can formally prove the asymptotic convergence of necessary and
sufficient conditions,
{and extend the bounds on the sparsity
factor of the sensing matrix in presence of communication noise}.
The asymptotic analysis under MAP decoding, 
both
for the noiseless case and in presence of communication noise \emph{b}),
are compatible with the results derived in \cite{CS}, 
{respectively}
under $\ell_{0}$-norm minimization decoding 
{and under
minimum-empirical entropy decoding}. Error exponents for MAP decoding
of correlated sources in the noiseless setting are compatible with
the ones presented in \cite{giannakis}, and are here extended to the case
of arbitrary statistical structure, and presence of noise contamination
both preceding and following the sensing operation. 

The rest of the paper is organized as follows. Section~\ref{sec:SystModel}
introduces the considered signal models 
{in the
context of data dissemination in a wireless sensor network}. In Section~\ref{sec:NecessaryCondition},
we derive the necessary conditions for asymptotic almost surely
exact recovery, both for the noiseless and noisy cases. Section~\ref{sec:SufficientConditionNC}
describes the sufficient conditions and the error exponents under
MAP decoding, for the noiseless case and in presence of communication
noise only. In Section~\ref{sec:SufficientConditionNS}, sensing
noise is also taken into account. Section~\ref{sec:Conclusions}
concludes the paper.

\section{System Model and Problem Setup}

\label{sec:SystModel}

This section introduces the system model as well as various hypotheses
on the sources and on the sensing and communication noises. In what
follows, sans-serif font denotes random quantities while serif font
denotes deterministic quantities. Matrices are in bold-face upper-case
letters. A length $n$ vector is in bold-face lower-case with a superscript
$n$. Calligraphic font denotes set, except $\mathcal{H}$, which
denotes the entropy rate. All logarithms are in base 2.

\subsection{The source model}

Consider a wireless sensor network consisting of a set $\mathcal{N}$
of $N=\left|\mathcal{N}\right|$ sensors. 
The target physical phenomenon (\emph{e.g.} the temperature) at the $n$-th sensor is represented by the random variable
$\mathsf\Theta_n$, taking values on a finite field $\mathbb F_Q$ of size
$Q$. Let $\boldsymbol \theta^N$ be a realization of the random vector $\boldsymbol{\mathsf\Theta}^N =
(\mathsf\Theta_1,\ldots,\mathsf\Theta_N)$, taking values in $\mathbb
F_Q^N$. The vector $\boldsymbol{\mathsf\Theta}^N$ represents the source in the
Bayesian compressed sensing framework. The probability
mass function (pmf) associated with $\boldsymbol{\mathsf\Theta}^{N}$ is denoted by $p\left(\boldsymbol{\theta}^{N}\right)$,
rather than
$p_{\boldsymbol{\mathsf{\Theta}}^{N}}\left(\boldsymbol{\theta}^{N}\right)$,
for the sake of simplicity. 
In general, the analytic form of $p\left(\boldsymbol{\theta}^{N}\right)$
depends on the characteristics of the observed phenomenon and of the
topology of the sensor network. Here we consider three different models, defined
as follows.
%
%

\textit{SI: Sparse, Independent and identically distributed source}.
Each element of the source vector $\boldsymbol{\mathsf{\Theta}}^{N}$ 
is \textit{\emph{independent and identically
distributed (iid)}} with pmf \textit{$p_{\mathsf{\Theta}}\left(\cdot\right)$
}and \textit{$p_{\mathsf{\Theta}}\left(0\right)>0.5$}, 
\begin{equation}
p\left(\boldsymbol{\theta}^{N}\right)=\prod_{n=1}^{N}p_{\mathsf{\Theta}}\left(\theta_{n}\right).\label{eq:PMFSI}
\end{equation}

\textit{StM: Stationary Markov model}. Let $\boldsymbol{\theta}_{n}^{n+r-1}\in\mathbb{F}_{Q}^{r}$
denote the sequences $\left(\theta_{n},\ldots,\theta_{n+r-1}\right)$. This 
is the stationary $r$-th order Markov model with $r\in\mathbb{N}^{+}$
and $1\leq r\ll N$ and transition probability $p\left(\theta_{n+r}\mid\boldsymbol{\theta}_{n}^{n+r-1}\right)$.
The pmf of $\boldsymbol{\mathsf\Theta}^{N}$ may be written as 
\begin{equation}
p\left(\boldsymbol{\theta}^{N}\right)=p\left(\boldsymbol{\theta}_{1}^{r}\right)\prod_{n=1}^{N-r}p\left(\theta_{n+r}\mid\boldsymbol{\theta}_{n}^{n+r-1}\right).\label{eq:PMFStM}
\end{equation}

\textit{GSE: General Stationary and Ergodic }\emph{model}. This is
the general case, without any further assumption apart from the ergodicity of
the source.

\subsection{The sensing model}

The considered system model is shown in Figure~\ref{fig:ModelSysteme}.
\begin{figure}
\begin{centering}
\includegraphics[width=0.65\columnwidth]{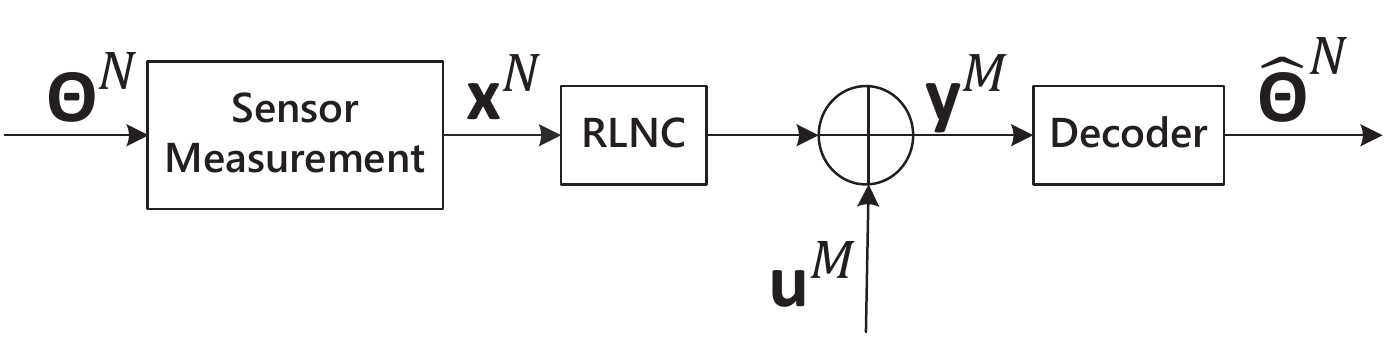}
\par\end{centering}
\caption{Block diagram for network compressive sensing model\label{fig:ModelSysteme}}
\end{figure}
Let $\mathsf{x}_{n}\in\mathbb{F}_{Q}$ be the measurement of $\mathsf{\Theta}_{n}$
obtained by the $n$-th sensor. The random vector $\boldsymbol{\mathsf{x}}^{N}=\left(\mathsf{x}_{1},\mathsf{x}_{2},\ldots,\mathsf{x}_{N}\right)\in\mathbb{F}_{Q}^{N}$
is a copy of the source vector $\boldsymbol{\mathsf{\Theta}}^{N}$ corrupted by
the \emph{sensing noise}.
The sensing noise models the effect of imperfect measure acquisition at each
sensor. It is described by 
the stationary transition probability
$p_{\mathsf{x}\mid\mathsf{\Theta}}(x_{n}\mid\theta_{n})$, $\forall n$.
Remark that this
implies that $\boldsymbol{\mathsf{x}}^{N}$ is stationary as
long as $\boldsymbol{\mathsf{\Theta}}^{N}$ is stationary.
The local measurement $\mathsf x_n$ at node $n$ is used to compute a
packet via RLNC \cite{RLNC}, which is then
broadcast and received by the neighbours of $n$.
Each node in the network can act as a sink, and attempt reconstruction of
$\boldsymbol{\mathsf\Theta}^N$, after a number
$M\leq N$ of linear combinations has been received.
The effects of RLNC at
a sink node can be modeled as multiplying $\boldsymbol{\mathsf{x}}^{N}$ by a random
matrix $\boldsymbol{\mathsf{A}}\in\mathbb{F}_{Q}^{M\times N}$.
We assume that
some \emph{communication noise} $\boldsymbol{\mathsf{u}}^{M}\in\mathbb{F}_{Q}^{M}$
affects the received packets, modeling the effects of transmission. Each entry of $\boldsymbol{\mathsf{u}}^{M}$
is iid with pmf $p_{\mathsf{u}}\left(\cdot\right)$. 
The sink node is assumed to have
received $M$ packets, with the $i$-th packet carrying the coefficients
$\boldsymbol{\mathsf{A}}_{i}$ and the result of linear combination
$\mathsf{y}_{i}\in\mathbb{F}_{Q}$, where $\boldsymbol{\mathsf{A}}_{i}$
is the $i$-th row of $\boldsymbol{\mathsf{A}}$ and $\mathsf{y}_{i}=\boldsymbol{\mathsf{A}}_{i}\boldsymbol{\mathsf{x}}^{N}+\mathsf{u}_{i}$,
with all operations in $\mathbb{F}_{Q}$. The vector $\boldsymbol{\mathsf{y}}^{M}=\left(\mathsf{y}_{1},\mathsf{y}_{2},\ldots,\mathsf{y}_{M}\right)^{\textrm{t}}\in\mathbb{F}_{Q}^{M}$
can be then represented as 
\begin{equation}
\boldsymbol{\mathsf{y}}^{M}=\boldsymbol{\mathsf{A}}\boldsymbol{\mathsf{x}}^{N}+\boldsymbol{\mathsf{u}}^{M},\label{eq:MesModel}
\end{equation}
where the network coding matrix $\boldsymbol{\mathsf A}$ plays the role of the
random sensing matrix in the compressed sensing setup.
According to the presence of the sensing and communication noises,
one obtains four types of noise models, namely Without Noise (\emph{WN}),
Noise in Communications (\emph{NC}) only, Noise in the Sensing process
(\emph{NS}) only, and noise in both Communications and in the Sensing
process (\emph{NCS}). These models are summarized in Table~\ref{tab:NoiseModel}.

\begin{table}[th]
\caption{Classification and notation based on the presence of noise\label{tab:NoiseModel}}
\medskip{}
\centering{}%
\begin{tabular}{c|c|c|c}
\multicolumn{1}{c}{} & \multicolumn{1}{c}{} & \multicolumn{2}{c}{Communication Noise}\tabularnewline
\cline{3-4} 
\multicolumn{1}{c}{} & \multicolumn{1}{c}{} & $\;\;\,$absent$\,\;\;$ & present\tabularnewline
\cline{3-4} 
Sensing & absent & WN & NC\tabularnewline
\cline{2-4} 
Noise & present & NS & NCS\tabularnewline
\end{tabular}
\end{table}

In
general, the matrix $\boldsymbol{\mathsf{A}}$ is not necessarily of full rank,
and it is assumed to be independent of $\boldsymbol{\mathsf{x}}^{N}$. Two
different assumptions about the structure of $\boldsymbol{\mathsf{A}}$
are considered here: (\emph{A1}) the entries of $\boldsymbol{\mathsf{A}}$ are iid,
uniformly distributed in $\mathbb{F}_{Q}$; (\emph{A2}) the entries of
$\boldsymbol{\mathsf A}$ are iid, all non-zero elements
of $\mathbb{F}_{Q}$ are equiprobable. Both can be represented
using the following model: for the entry $\mathsf{A}_{ij}$
of $\boldsymbol{\mathsf{A}}$, 
\begin{equation}
\Pr\left(\mathsf{A}_{ij}=q\right)=\begin{cases}
1-\gamma & q=0,\\
\gamma/\left(Q-1\right) & q\in\mathbb{F}_{Q}\setminus\left\{ 0\right\} ,
\end{cases}\label{eq:ASparistyModel}
\end{equation}
where $\gamma$ is the sparsity factor, $0<\gamma<1$, and $\boldsymbol{\mathsf{A}}$
is sparse if $\gamma<0.5$. 
We only assume that 
\begin{equation}
0<\gamma\leq1-Q^{-1}.\label{eq:gama}
\end{equation}
Notice that choosing $\gamma<1-Q^{-1}$ corresponds to assumption \emph{(A2)},
while choosing $\gamma=1-Q^{-1}$ corresponds to assumption \emph{(A1)}, since
(\ref{eq:ASparistyModel}) becomes
the uniform distribution.

In practice, sparse matrices are
preferable. As the information of the sensing matrix is carried in
the headers of packets \cite{PNC,CompHeader}, the network coding overhead may
be large if $\boldsymbol{\mathsf{A}}$ is dense and $N$ is large.
Moreover, as mentioned in \cite{EUSI}, sparse matrices facilitate
the convergence of the approximate belief propagation algorithm \cite{BP}.
 In practice, the structure of $\boldsymbol{\mathsf{A}}$ is strongly
dependent on the structure of the network. For example, \cite{giannakis}
assumes that only a subset of sensors $\mathcal{S}{}_{i}\subset\mathcal{N}$
have participated in the $i$-th linear mixing. The content of the
subsets $\mathcal{S}{}_{i}$ depends on the location of each sensor
and is designed to minimize communication costs. In $\boldsymbol{\mathsf{A}}$,
coefficients associated to nodes belonging to $\mathcal{S}{}_{i}$
follow a uniform distribution, while the others are null. This model,
however, is not considered here, since we aim at a general asymptotic
analysis, independent on the topology of the network.\textcolor{red}{{} }

\subsection{MAP Decoding}

The sink node observes the realization $\mathbf{y}^{M}$ and perfectly knows 
the realization
$\mathbf{A}$, \emph{e.g.}, from packet headers, see \cite{PNC}
and \cite{CompHeader}. The maximum \emph{a posteriori} estimate
$\hat{\boldsymbol\theta}^N$ of the realization of 
$\boldsymbol{\mathsf\Theta}^N$ is evaluated as
\begin{equation}
\hat{\boldsymbol{\theta}}^{N}=\textrm{arg}\max_{\boldsymbol{\theta}^{N}\in\mathbb{F}_{Q}^{N}}p\left(\boldsymbol{\theta}^{N}\mid\mathbf{y}^{M},\mathbf{A}\right),\label{eq:MAPEstimateDef}
\end{equation}
where the \emph{a posteriori} pmf is
\begin{align}
p\left(\boldsymbol{\theta}^{N}\mid\mathbf{y}^{M},\mathbf{A}\right) & \propto  p\left(\boldsymbol{\theta}^{N},\mathbf{y}^{M},\mathbf{A}\right)\nonumber \\
 & =  \sum_{\mathbf{x}^{N}\in\mathbb{F}_{Q}^{N}}\sum_{\mathbf{u}^{M}\in\mathbb{F}_{Q}^{M}}p\left(\boldsymbol{\theta}^{N},\mathbf{x}^{N},\mathbf{u}^{M},\mathbf{y}^{M},\mathbf{A}\right)\nonumber \\
 & =  \sum_{\mathbf{x}^{N}\in\mathbb{F}_{Q}^{N}}\sum_{\mathbf{u}^{M}\in\mathbb{F}_{Q}^{M}}p\left(\boldsymbol{\theta}^{N}\right)p\left(\mathbf{x}^{N}\mid\boldsymbol{\theta}^{N}\right)p\left(\mathbf{u}^{M}\right)p\left(\mathbf{A}\right)p\left(\mathbf{y}^{M}\mid\mathbf{x}^{N},\mathbf{u}^{M},\mathbf{A}\right).\label{eq:MAPEstimatePMF}
\end{align}
Note that the conditional pmf $p\left(\mathbf{y}^{M}\mid\mathbf{x}^{N},\mathbf{u}^{M},\mathbf{A}\right)$
is an indicator function, \emph{i.e.}, 
\begin{equation}
p\left(\mathbf{y}^{M}\mid\mathbf{x}^{N},\mathbf{u}^{M},\mathbf{A}\right)=1_{\mathbf{y}^{M}=\mathbf{A}\mathbf{x}^{N}+\mathbf{u}^{M}}.
\end{equation}
An error event (decoding error) occurs when $\hat{\boldsymbol{\theta}}^{N}\neq\boldsymbol{\theta}^{N}$,
with probability
\begin{equation}
P_{\text{e}}=\Pr\left\{ \hat{\boldsymbol{\mathsf{\Theta}}}^{N}\neq\boldsymbol{\mathsf{\Theta}}^{N}\right\} .\label{eq:ErrorEventProbDef}
\end{equation}

Our objective is to evaluate lower and upper bounds of (\ref{eq:ErrorEventProbDef})
under MAP decoding, as functions of $M$, $N$, and $\gamma$, for the
various source and noise models previously introduced. With these
bounds, one can obtain necessary and sufficient conditions on the
ratio $M/N$ for asymptotic (with $N\rightarrow\infty$) perfect recovery,
\emph{i.e.}, to obtain 
\begin{equation}
\lim_{N\rightarrow\infty}P_{\text{e}}=0.\label{eq:ErrorProbaAsymptotic}
\end{equation}

\section{Necessary Condition for Asymptotic Perfect Recovery}

\label{sec:NecessaryCondition}

This section derives the necessary conditions for asymptotically ($N\rightarrow\infty$)
vanishing probability of decoding error. They only depend on the assumptions
considered about the sensing and communication noises. We directly
analyze the\emph{ }\textit{\emph{NCS}}\emph{ }case for the GSE source
model. The results for this case can be easily adapted to the other
cases. This work extends results obtained in \cite{seong} for the
noiseless case (\textit{\emph{WN}}). Two situations are considered,
depending on the value of the entropy rate
\begin{equation}
\mathcal{H}\left(\mathsf{x}\right)=\lim_{N\rightarrow\infty}\frac{1}{N}H\left(\boldsymbol{\mathsf{x}}^{N}\right).\label{eq:EntropyRate}
\end{equation}

\begin{prop}[Necessary condition for the NCS
case]
\noindent \label{Prop:CNforNCSCase}
Assume the presence of both communication and sensing noises
and that $\mathcal{H}\left(\mathsf{x}\right)>0$. Consider some arbitrary
small $\delta\in\mathbb{R}^{+}$. For $N\rightarrow\infty$, the necessary
conditions for $P_{\text{e}}<\delta$ are 
\begin{equation}
\mathcal{H}\left(\mathsf{\Theta},\mathsf{x}\right)-\mathcal{H}\left(\mathsf{x}\right)<3\varepsilon+\delta\log Q,\label{eq:necessary1}
\end{equation}
\begin{equation}
H\left(p_{\mathsf{u}}\right)<\log Q,\label{eq:necessary2}
\end{equation}
and 
\begin{equation}
\frac{M}{N}>\frac{\mathcal{H}\left(\mathsf{\Theta},\mathsf{x}\right)-\left(5\varepsilon+2\delta\log Q\right)}{\log Q-H\left(p_{\mathsf{u}}\right)},\label{eq:necessary3}
\end{equation}
where $\varepsilon\in\mathbb{R}^{+}$ is an arbitrary small constant.\end{prop}
\begin{cor}
\label{Cor:CNforNCSCase}Consider the same hypotheses as in Proposition~\ref{Prop:CNforNCSCase}
and assume now that $\mathcal{H}\left(\mathsf{x}\right)=0$. Consider
some arbitrary small $\delta\in\mathbb{R}^{+}$. For $N\rightarrow\infty$,
the necessary condition for $P_{\text{e}}<\delta$ is 
\begin{equation}
\mathcal{H}\left(\mathsf{\Theta},\mathsf{x}\right)<3\varepsilon+\delta\log Q,\label{eq:necessary1-1}
\end{equation}
where $\varepsilon\in\mathbb{R}^{+}$ is an arbitrary small constant. 
\end{cor}
In Proposition~\ref{Prop:CNforNCSCase}, (\ref{eq:necessary1}) implies
that for asymptotically exact recovery, $p\left(\boldsymbol{\mathsf{x}}^{N}\mid\boldsymbol{\mathsf{\Theta}}^{N}\right)$
should degenerate, almost surely, into a deterministic mapping. The
condition (\ref{eq:necessary2}) indicates that asymptotically exact
recovery for non-deterministic sources is not possible in case of
uniformly distributed communication noise. Finally, (\ref{eq:necessary3})
indicates that the minimum number of required measurements depends
both on the sensing and communication noises as well as on the distribution
of $\mathsf{\Theta}$. In particular, for a given source with entropy
rate $\mathcal{H}\left(\mathsf{\Theta}\right)$, the number of necessary
measurements increases with the level of the sensing noise, determined
by $\mathcal{H}\left(\mathsf{x}\mid\mathsf{\Theta}\right)$. Similarly,
the number of necessary measurements increases when the communication
noise gets closer to uniformly distributed. The following proof is inspired by the work in \cite{seong}, 
with both communication noise and sensing noise are considered here.  
\begin{proof}
From the problem setup, one has the Markov chain 
\begin{equation}
\boldsymbol{\mathsf{\Theta}}^{N}\leftrightarrow\boldsymbol{\mathsf{x}}^{N}\leftrightarrow\left(\boldsymbol{\mathsf{y}}^{M},\boldsymbol{\mathsf{A}}\right)\leftrightarrow\hat{\boldsymbol{\mathsf{\Theta}}}^{N},\label{eq:MarcovChain}
\end{equation}
from which one deduces that
\begin{equation}
H\left(\boldsymbol{\mathsf{\Theta}}^{N}\mid\boldsymbol{\mathsf{x}}^{N}\right)\leq H\left(\boldsymbol{\mathsf{\Theta}}^{N}\mid\hat{\boldsymbol{\mathsf{\Theta}}}^{N}\right),\label{eq:constraint_a}
\end{equation}
and
\begin{equation}
H\left(\boldsymbol{\mathsf{x}}^{N}\mid\boldsymbol{\mathsf{y}}^{M},\boldsymbol{\mathsf{A}}\right)\leq H\left(\boldsymbol{\mathsf{\Theta}}^{N}\mid\hat{\boldsymbol{\mathsf{\Theta}}}^{N}\right).\label{eq:markov1}
\end{equation}
Applying Fano's inequality \cite[Sec. 2.10]{EIT}, one gets 
\begin{eqnarray}
H\left(\boldsymbol{\mathsf{\Theta}}^{N}\mid\hat{\boldsymbol{\mathsf{\Theta}}}^{N}\right) & \leq & 1+P_{\text{e}}\cdot\log\left(Q^{N}-1\right)\nonumber \\
 & < & 1+NP_{\text{e}}\log Q,\label{eq:fano}
\end{eqnarray}
an upper bound of $P_{\text{e}}$ is obtained combining (\ref{eq:constraint_a})
and (\ref{eq:fano}), 
\begin{equation}
P_{\text{e}}>\frac{H\left(\boldsymbol{\mathsf{\Theta}}^{N},\boldsymbol{\mathsf{x}}^{N}\right)-H\left(\boldsymbol{\mathsf{x}}^{N}\right)-1}{N\log Q}.\label{eq:lower1}
\end{equation}
Since $\boldsymbol{\mathsf{\Theta}}^{N}$ and $\boldsymbol{\mathsf{x}}^{N}$
are stationary and ergodic, for any $\varepsilon>0$, there exists
$N_{0}\in\mathbb{N}$ such that $\forall N>N_{0}$, one has 
\begin{equation}
\begin{cases}
\mathcal{H}\left(\mathsf{\Theta},\mathsf{x}\right)-\varepsilon
<
\frac{H\left(\boldsymbol{\mathsf{\Theta}}^{N},\boldsymbol{\mathsf{x}}^{N}\right)}{N}
<
\mathcal{H}\left(\mathsf{\Theta},\mathsf{x}\right)+\varepsilon
,\\
\mathcal{H}\left(\mathsf{x}\right)-\varepsilon
<
\frac{H\left(\boldsymbol{\mathsf{x}}^{N}\right)}{N}
<
\mathcal{H}\left(\mathsf{x}\right)+\varepsilon
,\\
\varepsilon>\frac{1}{N}.
\end{cases}\label{eq:lie}
\end{equation}
Hence for $N>N_{0}$, (\ref{eq:lower1}) can be rewritten as
\begin{equation}
P_{\text{e}}>\frac{\mathcal{H}\left(\mathsf{\Theta},\mathsf{x}\right)-\mathcal{H}\left(\mathsf{x}\right)-3\varepsilon}{\log Q}.\label{eq:lower2}
\end{equation}
For $P_{\text{e}}<\delta$, one deduces (\ref{eq:necessary1})
from (\ref{eq:lower2}). For $\delta$ and $\varepsilon$ arbitrary
small, (\ref{eq:necessary1}) imposes $\mathcal{H}\left(\mathsf{\Theta},\mathsf{x}\right)=\mathcal{H}\left(\mathsf{x}\right)$,
meaning that $\mathsf{\mathsf{\Theta}}$ should be deterministic knowing
$\mathsf{x}$, almost surely.

From (\ref{eq:markov1}) and (\ref{eq:fano}), one gets an other lower
bound for $P_{\text{e}}$ 
\begin{equation}
P_{\text{e}}>\frac{H\left(\boldsymbol{\mathsf{x}}^{N}\mid\boldsymbol{\mathsf{y}}^{M},\boldsymbol{\mathsf{A}}\right)-1}{N\log Q}.\label{eq:lower3}
\end{equation}
The conditional entropy $H\left(\boldsymbol{\mathsf{x}}^{N}\mid\boldsymbol{\mathsf{y}}^{M},\boldsymbol{\mathsf{A}}\right)$
can be bounded as 
\begin{eqnarray}
H\left(\boldsymbol{\mathsf{x}}^{N}\mid\boldsymbol{\mathsf{y}}^{M},\boldsymbol{\mathsf{A}}\right) & = & H\left(\boldsymbol{\mathsf{x}}^{N}\right)-I\left(\boldsymbol{\mathsf{x}}^{N};\boldsymbol{\mathsf{y}}^{M},\boldsymbol{\mathsf{A}}\right) \nonumber \\ & = & H\left(\boldsymbol{\mathsf{x}}^{N}\right)-\left(I\left(\boldsymbol{\mathsf{x}}^{N};\boldsymbol{\mathsf{A}}\right)+I\left(\boldsymbol{\mathsf{x}}^{N};\boldsymbol{\mathsf{y}}^{M}\mid\boldsymbol{\mathsf{A}}\right)\right)\nonumber \\
 & \overset{(a)}{=} & H\left(\boldsymbol{\mathsf{x}}^{N}\right)-\left(H\left(\boldsymbol{\mathsf{y}}^{M}\mid\boldsymbol{\mathsf{A}}\right)-H\left(\boldsymbol{\mathsf{y}}^{M}\mid\boldsymbol{\mathsf{A}},\boldsymbol{\mathsf{x}}^{N}\right)\right)\nonumber \\
 & \overset{(b)}{\geq} & H\left(\boldsymbol{\mathsf{x}}^{N}\right)-M\cdot\log Q+H\left(\boldsymbol{\mathsf{y}}^{M}\mid\boldsymbol{\mathsf{A}},\boldsymbol{\mathsf{x}}^{N}\right)\nonumber \\
 & \overset{(c)}{=} & H\left(\boldsymbol{\mathsf{x}}^{N}\right)-M\cdot\log
 Q+MH\left(p_{\mathsf{u}}\right),\label{eq:HXYA}
\end{eqnarray}
where $(a)$ follows from the assumption that $\boldsymbol{\mathsf{x}}^{N}$
and $\boldsymbol{\mathsf{A}}$ are independent, $(b)$ comes from
$H\left(\boldsymbol{\mathsf{y}}^{M}\mid\boldsymbol{\mathsf{A}}\right)\leq H\left(\boldsymbol{\mathsf{y}}^{M}\right)\leq\log\left|\mathbb{F}_{Q}^{M}\right|=M\log Q$,
and $(c)$ is because 
\begin{equation}
H\left(\boldsymbol{\mathsf{y}}^{M}\mid\boldsymbol{\mathsf{A}},\boldsymbol{\mathsf{x}}^{N}\right)=H\left(\boldsymbol{\mathsf{A}}\boldsymbol{\mathsf{x}}^{N}+\boldsymbol{\mathsf{u}}^{M}\mid\boldsymbol{\mathsf{A}},\boldsymbol{\mathsf{x}}^{N}\right)=H\left(\boldsymbol{\mathsf{u}}^{M}\right)=MH\left(p_{\mathsf{u}}\right).\label{eq:HYAX}
\end{equation}
Using (\ref{eq:lower3}) and (\ref{eq:HXYA}), a second necessary
condition for $P_{\text{e}}<\delta$ is
\begin{equation}
\frac{H\left(\boldsymbol{\mathsf{x}}^{N}\right)-M\left(\log Q-H\left(p_{\mathsf{u}}\right)\right)-1}{N\log Q}<\delta.\label{eq:lower4}
\end{equation}
For $N>N_{0}$, using (\ref{eq:lie}) in (\ref{eq:lower4}) yields
\begin{equation}
\frac{\mathcal{H}\left(\mathsf{x}\right)-\frac{M}{N}\left(\log Q-H\left(p_{\mathsf{u}}\right)\right)-2\varepsilon}{\log Q}<\delta.\label{eq:lower5}
\end{equation}
Now consider two cases. In the first case, the communication noise
is assumed uniformly distributed, \emph{i.e.}, 
\begin{equation}
H\left(p_{\mathsf{u}}\right)=\log Q,\label{eq:ass1}
\end{equation}
the condition (\ref{eq:lower5}) becomes
\begin{equation}
\mathcal{H}\left(\mathsf{x}\right)<\delta\log Q+2\varepsilon.\label{eq:ass2}
\end{equation}
As $\delta$ can be made arbitrary small, (\ref{eq:ass2}) imposes
that, for uniform communication noise, asymptotically vanishing probability
of error is possible only if $\mathcal{H}\left(\mathsf{x}\right)$
is arbitrary close to zero. For non-degenerate cases, \emph{i.e.},
$\mathcal{H}\left(\mathsf{x}\right)>0$, one obtains the necessary
condition (\ref{eq:necessary2}). In this second case, a lower bound
of the compression ratio $M/N$ is obtained immediately from (\ref{eq:lower5}),
\begin{equation}
\frac{M}{N}>\frac{\mathcal{H}\left(\mathsf{x}\right)-\left(2\varepsilon+\delta\log Q\right)}{\log Q-H\left(p_{\mathsf{u}}\right)}.\label{eq:lower6}
\end{equation}
We can represent the condition (\ref{eq:lower6}) in terms of the
joint entropy rate $\mathcal{H}\left(\mathsf{\Theta},\mathsf{x}\right)$
by applying (\ref{eq:necessary1}). Then, one gets (\ref{eq:necessary3})
and Proposition~\ref{Prop:CNforNCSCase} is proved. 

Consider now $\mathcal{H}\left(\mathsf{x}\right)=0$, then (\ref{eq:lower5})
holds for any value of $M/N$, and for any $H\left(p_{\mathsf{u}}\right)\leq\log Q$,
since the left side of (\ref{eq:lower5}) is always negative. Hence,
(\ref{eq:necessary1}) is the only necessary condition for this case.
Corollary~\ref{Prop:CNforNCSCase} is also proved. 
\end{proof}
With the results of the NCS noise model, one may derive the necessary
conditions for the other models. If no sensing noise is considered,
\emph{i.e.}, $\boldsymbol{\mathsf{x}}^{N}=\boldsymbol{\mathsf{\Theta}}^{N}$,
one has $H\left(\boldsymbol{\mathsf{\Theta}}^{N}\mid\boldsymbol{\mathsf{x}}^{N}\right)=H\left(\boldsymbol{\mathsf{x}}^{N}\mid\boldsymbol{\mathsf{\Theta}}^{N}\right)=0$
and $H\left(\boldsymbol{\mathsf{\Theta}}^{N},\boldsymbol{\mathsf{x}}^{N}\right)=H\left(\boldsymbol{\mathsf{\Theta}}^{N}\right)$.
If communication noise is absent, \emph{i.e.}, $\boldsymbol{\mathsf{u}}^{M}=\mathbf{0}$,
$H\left(\boldsymbol{\mathsf{u}}^{M}\right)=0$. The necessary conditions
for asymptotically ($N\rightarrow\infty$) vanishing probability of
decoding error for each case are listed in Table 2.

\begin{table}[H]
\caption{Necessary conditions for asymptotic perfect recovery in noiseless
and noisy cases}

\bigskip{}

\centering{}%
\begin{tabular}{c|c}
Case & Necessary Condition ($\mathcal{H}\left(\mathsf{x}\right)>0$)\tabularnewline
\hline 
\multirow{2}{*}{\textit{\emph{WN}}} & \multirow{2}{*}{$\frac{M}{N}>\frac{\mathcal{H}\left(\mathsf{\Theta}\right)}{\log Q}$,
already obtained in \cite{seong},}\tabularnewline
 & \tabularnewline
\multirow{2}{*}{\textit{\emph{NC}}} & \multirow{2}{*}{$\frac{M}{N}>\frac{\mathcal{H}\left(\mathsf{\Theta}\right)}{\log Q-H\left(p_{\mathsf{u}}\right)}$
and $H\left(p_{\mathsf{u}}\right)<\log Q$,}\tabularnewline
 & \tabularnewline
\multirow{2}{*}{\textit{\emph{NS}}} & \multirow{2}{*}{$\frac{M}{N}>\frac{\mathcal{H}\left(\mathsf{\Theta},\mathsf{x}\right)}{\log Q}$
and $\mathcal{H}\left(\mathsf{\Theta}\mid\mathsf{x}\right)=0$,}\tabularnewline
 & \tabularnewline
\multirow{2}{*}{\textit{\emph{NCS}}} & \multirow{2}{*}{$\frac{M}{N}>\frac{\mathcal{H}\left(\mathsf{\Theta},\mathsf{x}\right)}{\log Q-H\left(p_{\mathsf{u}}\right)}$
and $H\left(p_{\mathsf{u}}\right)<\log Q$ and $\mathcal{H}\left(\mathsf{\Theta}\mid\mathsf{x}\right)=0$.}\tabularnewline
 & \tabularnewline
\end{tabular}
\end{table}

\section{Sufficient Condition in Absence of Sensing Noise}

\label{sec:SufficientConditionNC}

This section provides an upper bound of the error probability for
the MAP estimation problem in absence of sensing noise (the\textit{
}WN and NC cases). These two cases are considered simultaneously because
their proofs are similar. When the channel noise vanishes, the NC
case boils down to the WN case.

\subsection{Upper Bound of the Error Probability}

\label{sec:SufficientConditionNC-1}
\begin{prop}[Upper bound of $P_e$, WN and NC cases]
\noindent \label{Prop:UpperBoundNC}Under MAP decoding, the asymptotic
($N\rightarrow\infty$) probability of error in absence of sensing
noise can be upper bounded as
\begin{equation}
P_{\textrm{e}}\leq P_{1}\left(\alpha\right)+P_{2}\left(\alpha\right)+2\varepsilon,\label{eq:PeUpperNC}
\end{equation}
where $\varepsilon\in\mathbb{R}^{+}$ is an arbitrarily small constant.
$P_{1}\left(\alpha\right)$ and $P_{2}\left(\alpha\right)$ are defined
as 
\begin{equation}
P_{1}\left(\alpha\right)=2^{-N\left(-\frac{M}{N}\left(H\left(p_{\mathsf{u}}\right)+\log\left(1-\gamma\right)+\varepsilon\right)-H_{2}\left(\alpha\right)-\alpha\log\left(Q-1\right)-\frac{\log\left(\alpha N\right)}{N}\right)},\label{eq:P1alpha}
\end{equation}
and
\begin{equation}
P_{2}\left(\alpha\right)=2^{-N\left(-\mathcal{H}\left(\mathsf{\mathsf{\Theta}}\right)-\frac{M}{N}\left(H\left(p_{\mathsf{u}}\right)+\log\left(Q^{-1}+\left(1-\frac{\gamma}{1-Q^{-1}}\right)^{\left\lceil \alpha N\right\rceil }\left(1-Q^{-1}\right)\right)+\varepsilon\right)-\varepsilon\right)},\label{eq:P2alpha}
\end{equation}
with $\alpha\in\mathbb{R}^{+}$ and $\alpha<0.5$.\end{prop}
\begin{proof}
The proof consists of two parts. First we define the error event,
and then we analyze the probability of error. 

Since no sensing noise is considered, we have $\boldsymbol{\mathsf{x}}^{N}=\boldsymbol{\mathsf{\Theta}}^{N}$
throughout this section. The \textit{a posteriori} pmf (\ref{eq:MAPEstimatePMF})
becomes 
\begin{equation}
p\left(\boldsymbol{\theta}^{N}\mid\mathbf{y}^{M},\mathbf{A}\right)\propto\sum_{\mathbf{u}^{M}\in\mathbb{F}_{Q}^{M}}p\left(\boldsymbol{\theta}^{N}\right)p\left(\mathbf{u}^{M}\right)p\left(\mathbf{A}\right)1_{\mathbf{y}^{M}=\mathbf{A}\boldsymbol{\theta}^{N}+\mathbf{u}^{M}}.\label{eq:aposteriori_noiseless-1}
\end{equation}
Suppose that $\boldsymbol{\theta}^{N}$ (given but unknown) is the
true state vector and consider that $\mathbf{A}$ has been generated
randomly. At the sink, $\mathbf{A}$ and $\mathbf{y}^{M}$ are known.
With MAP decoding, the reconstruction $\hat{\boldsymbol{\theta}}^{N}$
in (\ref{eq:MAPEstimateDef}) is 
\begin{equation}
\hat{\boldsymbol{\theta}}^{N}=\textrm{arg}\max_{\boldsymbol{\theta}^{N}\in\mathbb{F}_{Q}^{N}}\sum_{\mathbf{u}^{M}\in\mathbb{F}_{Q}^{M}}p\left(\boldsymbol{\theta}^{N}\right)p\left(\mathbf{u}^{M}\right)p\left(\mathbf{A}\right)1_{\mathbf{y}^{M}=\mathbf{A}\boldsymbol{\theta}^{N}+\mathbf{u}^{M}}.\label{eq:ThetaHeadNC}
\end{equation}
A decoding error happens if there exists a vector $\mathbf{\boldsymbol{\varphi}}^{N}\in\mathbb{F}_{Q}^{N}\setminus\left\{ \boldsymbol{\theta}^{N}\right\} $
such that 
\begin{equation}
\sum_{\mathbf{v}^{M}\in\mathbb{F}_{Q}^{M}}p\left(\mathbf{\boldsymbol{\varphi}}^{N}\right)p\left(\mathbf{v}^{M}\right)1_{\mathbf{y}^{M}=\mathbf{A}\mathbf{\boldsymbol{\varphi}}^{N}+\mathbf{v}^{M}}\geq\sum_{\mathbf{u}^{M}\in\mathbb{F}_{Q}^{M}}p\left(\boldsymbol{\theta}^{N}\right)p\left(\mathbf{u}^{M}\right)1_{\mathbf{y}^{M}=\mathbf{A}\boldsymbol{\theta}^{N}+\mathbf{u}^{M}}.\label{eq:ErrEventNC1}
\end{equation}
For fixed $\mathbf{y}^{M}$, $\mathbf{A}$, and $\boldsymbol{\theta}^{N}$,
there is exactly one vector $\mathbf{u}^{M}$ such that $\mathbf{u}^{M}=\mathbf{y}^{M}-\mathbf{A}\boldsymbol{\theta}^{N}$.
Hence the right side of (\ref{eq:ErrEventNC1}) can be represented
as $p_{\boldsymbol{\mathsf{\Theta}}^{N}}\left(\boldsymbol{\theta}^{N}\right)p_{\boldsymbol{\mathsf{u}}^{M}}\left(\mathbf{y}^{M}-\mathbf{A}\boldsymbol{\theta}^{N}\right)$.
The subscripts for the pmfs are introduced to avoid any ambiguity
of notations. Then (\ref{eq:ErrEventNC1}) is equivalent to 
\begin{equation}
p_{\boldsymbol{\mathsf{\Theta}}^{N}}\left(\mathbf{\boldsymbol{\varphi}}^{N}\right)p_{\boldsymbol{\mathsf{u}}^{M}}\left(\mathbf{y}^{M}-\mathbf{A}\mathbf{\boldsymbol{\varphi}}^{N}\right)\geq p_{\boldsymbol{\mathsf{\Theta}}^{N}}\left(\boldsymbol{\theta}^{N}\right)p_{\boldsymbol{\mathsf{u}}^{M}}\left(\mathbf{y}^{M}-\mathbf{A}\boldsymbol{\theta}^{N}\right).\label{eq:ErrEventNC2}
\end{equation}
An alternative way to state the error event can be: For a given realization
$\boldsymbol{\mathsf{\Theta}}^{N}=\boldsymbol{\theta}^{N}$, which
implies the realization $\boldsymbol{\mathsf{u}}^{M}=\mathbf{u}^{M}=\mathbf{y}^{M}-\mathbf{A}\boldsymbol{\theta}^{N}$,
there exists a pair $(\mathbf{\boldsymbol{\varphi}}^{N},\mathbf{v}^{M})\in\mathbb{F}_{Q}^{N}\times\mathbb{F}_{Q}^{M}$
such that
\begin{equation}
\begin{cases}
\mathbf{\boldsymbol{\varphi}}^{N}\neq\boldsymbol{\theta}^{N},\\
\mathbf{A}\mathbf{\boldsymbol{\varphi}}^{N}+\mathbf{v}^{M}=\mathbf{y}^{M}=\mathbf{A}\boldsymbol{\theta}^{N}+\mathbf{u}^{M},\\
p\left(\mathbf{\boldsymbol{\varphi}}^{N}\right)p\left(\mathbf{v}^{M}\right)\geq p\left(\boldsymbol{\theta}^{N}\right)p\left(\mathbf{u}^{M}\right).
\end{cases}\label{eq:ErrEventNC3}
\end{equation}
From conditions (\ref{eq:ErrEventNC3}), one concludes that the MAP
decoder is equivalent to the maximum $Q$-probability decoder \cite{csiszar82}
in the NC case.

An upper bound of the error probability is now derived. For a fixed
$\boldsymbol{\theta}^{N}$ and $\mathbf{u}^{M}$, the conditional
error probability is denoted by $\Pr\left\{ \textrm{error}\mid\boldsymbol{\theta}^{N},\mathbf{u}^{M}\right\} $.
The average error probability is 
\begin{equation}
P_{\textrm{e}}=\sum_{\boldsymbol{\theta}^{N}\in\mathbb{F}_{Q}^{N}}\sum_{\mathbf{u}^{M}\in\mathbb{F}_{Q}^{M}}p\left(\boldsymbol{\theta}^{N},\mathbf{u}^{M}\right)\Pr\left\{ \textrm{error}\mid\boldsymbol{\theta}^{N},\mathbf{u}^{M}\right\} .\label{eq:PeNC1}
\end{equation}
Weak typicality is instrumental in the following proofs. The notations
of \cite[Definition 4.2]{FIT} are extended to stationary and ergodic
sources. For any positive real number $\varepsilon$ and some integer
$N>0$, the weakly typical set $\mathcal{A}_{\left[\mathsf{\Theta}\right]\varepsilon}^{N}\subset\mathbb{F}_{Q}^{N}$
for a stationary and ergodic source $\boldsymbol{\mathsf{\Theta}}^{N}$
is the set of vectors $\boldsymbol{\theta}^{N}\in\mathbb{F}_{Q}^{N}$
satisfying
\begin{equation}
\left|-\frac{1}{N}\log p\left(\boldsymbol{\theta}^{N}\right)-\mathcal{H}\left(\mathsf{\mathsf{\Theta}}\right)\right|\leq\varepsilon,\label{eq:TypicalTheta}
\end{equation}
where $\mathcal{H}\left(\mathsf{\mathsf{\Theta}}\right)$ is the entropy
rate of the source. Similarly, for the noise vector $\boldsymbol{\mathsf{u}}^{M}$,
define 
\begin{equation}
\mathcal{A}_{\left[\mathsf{u}\right]\varepsilon}^{M}=\left\{ \mathbf{u}^{M}\in\mathbb{F}_{Q}^{M}:\:\left|-\frac{1}{M}\log p\left(\mathbf{u}^{M}\right)-H\left(p_{\mathsf{u}}\right)\right|\leq\varepsilon\right\} .\label{eq:TypicalU}
\end{equation}
Recall that the entries of $\boldsymbol{\mathsf{u}}^{M}$ are uncorrelated,
so $\mathcal{H}\left(\mathsf{u}\right)=H\left(p_{\mathsf{u}}\right)$.
Thanks to Shannon-McMillan-Breiman theorem \cite[Sec. 16.8]{EIT},
the pmf of the general stationary and ergodic source converges. In
other words, for any $\varepsilon>0$, there exists $N_{\varepsilon}$
and $M_{\varepsilon}$ such that for all $N>N_{\varepsilon}$ and
$M>M_{\varepsilon}$,
\begin{equation}
\Pr\left\{ \left|-\frac{1}{N}\log p\left(\boldsymbol{\mathsf{\Theta}}^{N}\right)-\mathcal{H}\left(\mathsf{\Theta}\right)\right|\leq\varepsilon\right\} \geq1-\varepsilon,\label{eq:AEPTheta}
\end{equation}
and
\begin{equation}
\Pr\left\{ \left|-\frac{1}{M}\log p\left(\boldsymbol{\mathsf{u}}^{M}\right)-\mathcal{H}\left(p_{\mathsf{u}}\right)\right|\leq\varepsilon\right\} \geq1-\varepsilon.\label{eq:AEPU}
\end{equation}
We can make $\varepsilon$ arbitrary close to zero as $N\rightarrow\infty$
and $M\rightarrow\infty$. A sandwich proof of this theorem is proposed
in \cite[Sec. 16.8]{EIT}. For the sparse and uncorrelated source
as defined in (\ref{eq:PMFSI}), $\mathcal{H}\left(\mathsf{\Theta}\right)$
is equal to $H\left(p_{\mathsf{\Theta}}\right)$, the entropy of a
single source. The entropy rate of the \textit{\emph{StM}} source
is the conditional entropy $H\left(\mathsf{\Theta}_{n+r}\mid\mathbf{\boldsymbol{\mathsf{\Theta}}}_{n}^{n+r-1}\right)$. 

From (\ref{eq:AEPTheta}) and (\ref{eq:AEPU}), one has $\Pr\left\{ \boldsymbol{\mathsf{\Theta}}^{N}\in\mathcal{A}_{\left[\mathsf{\Theta}\right]\varepsilon}^{N}\right\} \geq1-\varepsilon$
and $\Pr\left\{ \boldsymbol{\mathsf{u}}^{M}\in\mathcal{A}_{\left[\mathsf{u}\right]\varepsilon}^{M}\right\} \geq1-\varepsilon$
for $N>N_{\varepsilon}$ and $M>M_{\varepsilon}$. It implies that,
for $N$ and $M$ sufficiently large, $\boldsymbol{\mathsf{\Theta}}^{N}$
and $\boldsymbol{\mathsf{u}}^{M}$ belong to the weakly typical set
$\mathcal{A}_{\left[\mathsf{\Theta}\right]\varepsilon}^{N}$ and $\mathcal{A}_{\left[\mathsf{u}\right]\varepsilon}^{M}$,
almost surely. With respect to the typicality, $\mathbb{F}_{Q}^{N}\times\mathbb{F}_{Q}^{M}$
can be divided into two parts. Define the sets $\mathcal{U}$ and
$\mathcal{U}^{\textrm{c}}$ for the pair of vectors $\left(\boldsymbol{\theta}^{N},\mathbf{u}^{M}\right)$,
such that $\mathcal{U}\cup\mathcal{U}^{\textrm{c}}=\mathbb{F}_{Q}^{N}\times\mathbb{F}_{Q}^{M}$
and 
\begin{equation}
\mathcal{U}=\left\{ \boldsymbol{\theta}^{N}\in\mathbb{F}_{Q}^{N},\:\mathbf{u}^{M}\in\mathbb{F}_{Q}^{M}:\:\boldsymbol{\theta}^{N}\in\mathcal{A}_{\left[\mathsf{\Theta}\right]\varepsilon}^{N}\textrm{ and }\mathbf{u}^{M}\in\mathcal{A}_{\left[\mathsf{u}\right]\varepsilon}^{M}\right\} ,\label{eq:U1}
\end{equation}
\begin{equation}
\mathcal{U}^{\textrm{c}}=\left\{ \boldsymbol{\theta}^{N}\in\mathbb{F}_{Q}^{N},\:\mathbf{u}^{M}\in\mathbb{F}_{Q}^{M}:\:\boldsymbol{\theta}^{N}\notin\mathcal{A}_{\left[\mathsf{\Theta}\right]\varepsilon}^{N}\textrm{ or }\mathbf{u}^{M}\notin\mathcal{A}_{\left[\mathsf{u}\right]\varepsilon}^{M}\right\} .\label{eq:UC}
\end{equation}
$\mathcal{U}$ is the joint typical set for $(\boldsymbol{\theta}^{N},\mathbf{u}^{M})$,
due to the independence of $\boldsymbol{\mathsf{\Theta}}^{N}$ and
$\boldsymbol{\mathsf{u}}^{M}$. The error probability can be bounded
as
\begin{eqnarray}
P_{\textrm{e}} & = & \sum_{\left(\boldsymbol{\theta}^{N},\mathbf{u}^{M}\right)\in\mathcal{U}}p\left(\boldsymbol{\theta}^{N}\right)p\left(\mathbf{u}^{M}\right)\cdot\Pr\left\{ \textrm{error}\mid\boldsymbol{\theta}^{N},\mathbf{u}^{M}\right\} \nonumber \\
&&+\sum_{\left(\boldsymbol{\theta}^{N},\mathbf{u}^{M}\right)\in\mathcal{U}^{\textrm{c}}}p\left(\boldsymbol{\theta}^{N}\right)p\left(\mathbf{u}^{M}\right)\cdot\Pr\left\{ \textrm{error}\mid\boldsymbol{\theta}^{N},\mathbf{u}^{M}\right\} \nonumber \\
 & \overset{(a)}{\leq} & \sum_{\left(\boldsymbol{\theta}^{N},\mathbf{u}^{M}\right)\in\mathcal{U}}p\left(\boldsymbol{\theta}^{N}\right)p\left(\mathbf{u}^{M}\right)\cdot\Pr\left\{ \textrm{error}\mid\boldsymbol{\theta}^{N},\mathbf{u}^{M}\right\} +\sum_{\left(\boldsymbol{\theta}^{N},\mathbf{u}^{M}\right)\in\mathcal{U}^{\textrm{c}}}p\left(\boldsymbol{\theta}^{N}\right)p\left(\mathbf{u}^{M}\right)\nonumber \\
 & \overset{(b)}{\leq} & \sum_{\left(\boldsymbol{\theta}^{N},\mathbf{u}^{M}\right)\in\mathcal{U}}p\left(\boldsymbol{\theta}^{N}\right)p\left(\mathbf{u}^{M}\right)\cdot\Pr\left\{ \textrm{error}\mid\boldsymbol{\theta}^{N},\mathbf{u}^{M}\right\} +2\varepsilon,\label{eq:PeNCUpp1}
\end{eqnarray}
where $(a)$ comes from $\Pr\left(\textrm{error}\mid\boldsymbol{\theta}^{N},\mathbf{u}^{M}\right)\leq1$
and $(b)$ follows from the fact that 
\begin{eqnarray}
\sum_{\left(\boldsymbol{\theta}^{N},\mathbf{u}^{M}\right)\in\mathcal{U}^{\textrm{c}}}p\left(\boldsymbol{\theta}^{N}\right)p\left(\mathbf{u}^{M}\right) & = & 1-\sum_{\left(\boldsymbol{\theta}^{N},\mathbf{u}^{M}\right)\in\mathcal{U}}p\left(\boldsymbol{\theta}^{N}\right)p\left(\mathbf{u}^{M}\right)\nonumber \\
 & = & 1-\sum_{\boldsymbol{\theta}^{N}\in\mathcal{A}_{\left[\mathsf{\Theta}\right]\varepsilon}^{N}}p\left(\boldsymbol{\theta}^{N}\right)\sum_{\mathbf{u}^{M}\in\mathcal{A}_{\left[\mathsf{u}\right]\varepsilon}^{M}}p\left(\mathbf{u}^{M}\right)\nonumber \\
 & \leq & 1-\left(1-\varepsilon\right)\left(1-\varepsilon\right)\leq2\varepsilon.\label{eq:PeNCUpp2}
\end{eqnarray}
Since $\bold{A}$ is generated randomly, define the
random event
\begin{equation}
\mathcal{E}\left(\boldsymbol{\theta}^{N},\mathbf{u}^{M};\mathbf{\boldsymbol{\varphi}}^{N},\mathbf{v}^{M}\right)=\left\{ \boldsymbol{\mathsf{A}}\boldsymbol{\theta}^{N}+\mathbf{u}^{M}=\boldsymbol{\mathsf{A}}\mathbf{\boldsymbol{\varphi}}^{N}+\mathbf{v}^{M}\right\} ,
\end{equation}
where $\left(\boldsymbol{\theta}^{N},\mathbf{u}^{M}\right)$ is the
realization of the environment state, and $(\mathbf{\boldsymbol{\varphi}}^{N},\mathbf{v}^{M})$
is the potential reconstruction result. Conditioned on $\left(\boldsymbol{\theta}^{N},\mathbf{u}^{M}\right)$,
$\Pr\left\{ \textrm{error}\mid\boldsymbol{\theta}^{N},\mathbf{u}^{M}\right\} $
is in fact the probability of the union of the events $\mathcal{E}\left(\boldsymbol{\theta}^{N},\mathbf{u}^{M};\mathbf{\boldsymbol{\varphi}}^{N},\mathbf{v}^{M}\right)$
with all the parameter pairs $(\mathbf{\boldsymbol{\varphi}}^{N},\mathbf{v}^{M})\in\mathbb{F}_{Q}^{N}\times\mathbb{F}_{Q}^{M}$
such that $\mathbf{\boldsymbol{\varphi}}^{N}\neq\boldsymbol{\theta}^{N}$
and $p\left(\mathbf{\boldsymbol{\varphi}}^{N}\right)p\left(\mathbf{v}^{M}\right)\geq p\left(\boldsymbol{\theta}^{N}\right)p\left(\mathbf{u}^{M}\right)$,
see (\ref{eq:ErrEventNC3}). The conditional error probability can
then be rewritten as
\begin{equation}
\Pr\left\{ \textrm{error}\mid\boldsymbol{\theta}^{N},\mathbf{u}^{M}\right\} =\Pr\left\{ \underset{_{p\left(\mathbf{\boldsymbol{\varphi}}^{N}\right)p\left(\mathbf{v}^{M}\right)\geq p\left(\boldsymbol{\theta}^{N}\right)p\left(\mathbf{u}^{M}\right)}^{\:\mathbf{\boldsymbol{\varphi}}^{N}\in\mathbb{F}_{Q}^{N}\setminus\left\{ \boldsymbol{\theta}^{N}\right\} ,\:\mathbf{v}^{M}\in\mathbb{F}_{Q}^{M}:}}{\bigcup}\mathcal{E}\left(\boldsymbol{\theta}^{N},\mathbf{u}^{M};\mathbf{\boldsymbol{\varphi}}^{N},\mathbf{v}^{M}\right)\right\} .\label{eq:PeConditional}
\end{equation}
Introducing (\ref{eq:PeConditional}) in (\ref{eq:PeNCUpp1}) and
applying the union bound yields 
\begin{eqnarray}
P_{\textrm{e}} & \leq & \sum_{\left(\boldsymbol{\theta}^{N},\mathbf{u}^{M}\right)\in\mathcal{U}}p\left(\boldsymbol{\theta}^{N}\right)p\left(\mathbf{u}^{M}\right)\underset{_{p\left(\mathbf{\boldsymbol{\varphi}}^{N}\right)p\left(\mathbf{v}^{M}\right)\geq p\left(\boldsymbol{\theta}^{N}\right)p\left(\mathbf{u}^{M}\right)}^{\:\mathbf{\boldsymbol{\varphi}}^{N}\in\mathbb{F}_{Q}^{N}\setminus\left\{ \boldsymbol{\theta}^{N}\right\} ,\:\mathbf{v}^{M}\in\mathbb{F}_{Q}^{M}:}}{\sum}\Pr\left\{ \mathcal{E}\left(\boldsymbol{\theta}^{N},\mathbf{u}^{M};\mathbf{\boldsymbol{\varphi}}^{N},\mathbf{v}^{M}\right)\right\} +2\varepsilon\nonumber \\
 & = &
 \sum_{\left(\boldsymbol{\theta}^{N},\mathbf{u}^{M}\right)\in\mathcal{U}}p\left(\boldsymbol{\theta}^{N}\right)p\left(\mathbf{u}^{M}\right)\nonumber \\
&&\cdot \underset{_{\quad\mathbf{v}^{M}\in\mathbb{F}_{Q}^{M}}^{\mathbf{\boldsymbol{\varphi}}^{N}\in\mathbb{F}_{Q}^{N}\setminus\left\{
 \boldsymbol{\theta}^{N}\right\} }}{\sum}\Phi\left(\boldsymbol{\theta}^{N},\mathbf{u}^{M};\mathbf{\boldsymbol{\varphi}}^{N},\mathbf{v}^{M}\right)\Pr\left\{ \mathcal{E}\left(\boldsymbol{\theta}^{N},\mathbf{u}^{M};\mathbf{\boldsymbol{\varphi}}^{N},\mathbf{v}^{M}\right)\right\} +2\varepsilon,\label{eq:PeNCUpp3}
\end{eqnarray}
where 
\begin{equation}
\Phi\left(\boldsymbol{\theta}^{N},\mathbf{u}^{M};\mathbf{\boldsymbol{\varphi}}^{N},\mathbf{v}^{M}\right)=\begin{cases}
1 & \textrm{ if }p\left(\mathbf{\boldsymbol{\varphi}}^{N}\right)p\left(\mathbf{v}^{M}\right)\geq p\left(\boldsymbol{\theta}^{N}\right)p\left(\mathbf{u}^{M}\right),\\
0 & \textrm{ if }p\left(\mathbf{\boldsymbol{\varphi}}^{N}\right)p\left(\mathbf{v}^{M}\right)<p\left(\boldsymbol{\theta}^{N}\right)p\left(\mathbf{u}^{M}\right).
\end{cases}\label{eq:IndicatorNC}
\end{equation}
Now consider the following lemma.

\medskip{}
\noindent \textbf{Lemma 1} (Upper bound of $\Phi$). \textit{Consider
some $s\in\mathbb{R}^{+}$ with $s\leq1$. For any $\boldsymbol{\theta}^{N},\mathbf{\boldsymbol{\varphi}}^{N}$
in }$\mathbb{F}_{Q}^{N}$\textit{ and $\mathbf{u}^{M},\mathbf{v}^{M}$
in $\mathbb{F}_{Q}^{M}$, the following inequality holds},
\begin{equation}
\Phi\left(\boldsymbol{\theta}^{N},\mathbf{\boldsymbol{\varphi}}^{N},\mathbf{u}^{M},\mathbf{v}^{M}\right)\leq\left(\frac{p\left(\mathbf{\boldsymbol{\varphi}}^{N}\right)p\left(\mathbf{v}^{M}\right)}{p\left(\boldsymbol{\theta}^{N}\right)p\left(\mathbf{u}^{M}\right)}\right)^{s}.\label{eq:BoundGallager}
\end{equation}

Lemma 1 is a part of Gallager's derivation of error exponents in \cite[Sec. 5.6]{gallager}.
Introducing (\ref{eq:BoundGallager}) with $s=1$ into (\ref{eq:PeNCUpp3}),
one gets 
\begin{equation}
P_{\textrm{e}}\leq\sum_{\left(\boldsymbol{\theta}^{N},\mathbf{u}^{M}\right)\in\mathcal{U}}\underset{_{\quad\mathbf{v}^{M}\in\mathbb{F}_{Q}^{M}}^{\mathbf{\boldsymbol{\varphi}}^{N}\in\mathbb{F}_{Q}^{N}\setminus\left\{ \boldsymbol{\theta}^{N}\right\} }}{\sum}p\left(\mathbf{\boldsymbol{\varphi}}^{N}\right)p\left(\mathbf{v}^{M}\right)\Pr\left\{ \mathcal{E}\left(\boldsymbol{\theta}^{N},\mathbf{u}^{M};\mathbf{\boldsymbol{\varphi}}^{N},\mathbf{v}^{M}\right)\right\} +2\varepsilon.\label{eq:PeNCUpp4}
\end{equation}
In (\ref{eq:PeNCUpp4}), 
\begin{equation}
\Pr\left\{ \mathcal{E}\left(\boldsymbol{\theta}^{N},\mathbf{u}^{M};\mathbf{\boldsymbol{\varphi}}^{N},\mathbf{v}^{M}\right)\right\} 
=\Pr\left\{ \boldsymbol{\mathsf{A}}\mathbf{\boldsymbol{\mu}}^{N}=\mathbf{s}^{M}\mid\mathbf{\boldsymbol{\mu}}^{N}\neq\mathbf{0},\mathbf{s}^{M}\right\}
\end{equation}
with $\mathbf{\boldsymbol{\mu}}^{N}=\mathbf{\boldsymbol{\varphi}}^{N}-\boldsymbol{\theta}^{N}\in\mathbb{F}_{Q}^{N}\setminus\left\{ \mathbf{0}\right\} $,
and $\mathbf{s}^{M}=\mathbf{v}^{M}-\mathbf{u}^{M}\in\mathbb{F}_{Q}^{M}$.
This probability depends on the sparsity of $\mathbf{\boldsymbol{\mu}}^{N}$
and of $\mathbf{s}^{M}$, let $d_{1}=\left\Vert \mathbf{\boldsymbol{\mu}}^{N}\right\Vert _{0}$
and $d_{2}=\left\Vert \mathbf{s}^{M}\right\Vert _{0}$. Both $d_{1}$
and $d_{2}$ are integers such that $1\leq d_{1}\leq N$ and $0\leq d_{2}\leq M$.
Define the multivariable function 
\begin{equation}
f\left(d_{1},d_{2};\gamma,Q,M\right)=\Pr\left\{ \boldsymbol{\mathsf{A}}\mathbf{\boldsymbol{\mu}}^{N}=\mathbf{s}^{M}\mid\left\Vert \mathbf{\boldsymbol{\mu}}^{N}\right\Vert _{0}=d_{1},\left\Vert \mathbf{s}^{M}\right\Vert _{0}=d_{2}\right\} ,\label{eq:fd1d2}
\end{equation}
where $\gamma$, $Q$, and $M$ are the parameters of the random matrix
$\boldsymbol{\mathsf{A}}$. 
\begin{equation}
\Pr\left\{ \boldsymbol{\mathsf{A}}\mathbf{\boldsymbol{\mu}}^{N}=\mathbf{0}\mid\mathbf{\boldsymbol{\mu}}^{N}\neq\mathbf{0}\right\} =f\left(d_{1},0;\gamma,Q,M\right) \nonumber
\end{equation}
has been evaluated in \cite{rankmin} and \cite{seong}. We provide
a simple extension of this result for $d_{2}\neq0$.
\noindent \medskip{}

\noindent \textbf{Lemma 2} (Properties of $f\left(d_{1},d_{2};\gamma,Q,M\right)$).
\textit{The function $f\left(d_{1},d_{2};\gamma,Q,M\right)$, defined
in (\ref{eq:fd1d2}), is non-increasing in $d_{2}$ for a given $d_{1}$
and 
\begin{equation}
f\left(d_{1},d_{2};\gamma,Q,M\right)\leq f\left(d_{1},0;\gamma,Q,M\right)=\left(Q^{-1}+\left(1-\frac{\gamma}{1-Q^{-1}}\right)^{d_{1}}\left(1-Q^{-1}\right)\right)^{M}.\label{eq:fd1d2Up}
\end{equation}
Moreover $f\left(d_{1},0;\gamma,Q,M\right)$ is non-increasing in $d_{1}$
and 
\begin{equation}
f\left(d_{1},0;\gamma,Q,M\right)\leq f\left(1,0;\gamma,Q,M\right)=\left(1-\gamma\right)^{M}.\label{eq:fd10}
\end{equation}
If $\gamma=1-Q^{-1}$, which corresponds to a uniformly distributed
sensing matrix, 
\begin{equation}
f\left(d_{1},d_{2};\gamma,Q,M\right)=Q^{-M}
\end{equation}
is
constant.\bigskip{}
}

\noindent See Appendix~\ref{sec:Proof_lemma1} for the proof details.
Using Lemma 2, (\ref{eq:PeNCUpp4}) can be expressed as 
\begin{eqnarray}
P_{\textrm{e}} & \overset{(a)}{\leq} & \sum_{d_{1}=1}^{N}\sum_{d_{2}=0}^{M}\underset{_{\mathbf{v}^{M}\in\mathbb{F}_{Q}^{M}:\left\Vert \mathbf{u}^{M}-\mathbf{v}^{M}\right\Vert _{0}=d_{2}}^{_{\,\mathbf{\boldsymbol{\varphi}}^{N}\in\mathbb{F}_{Q}^{N}:\left\Vert \mathbf{\boldsymbol{\varphi}}^{N}-\boldsymbol{\theta}^{N}\right\Vert _{0}=d_{1}}^{\qquad\left(\boldsymbol{\theta}^{N},\mathbf{u}^{M}\right)\in\mathcal{U}}}}{\sum}p\left(\mathbf{\boldsymbol{\varphi}}^{N}\right)p\left(\mathbf{v}^{M}\right)f\left(d_{1},d_{2};\gamma,Q,M\right)+2\varepsilon\nonumber \\
 & \overset{(b)}{\leq} & \sum_{d_{1}=1}^{N}\underset{^{_{\mathbf{\boldsymbol{\varphi}}^{N}\in\mathbb{F}_{Q}^{N}:\left\Vert \mathbf{\boldsymbol{\varphi}}^{N}-\boldsymbol{\theta}^{N}\right\Vert _{0}=d_{1}}^{\qquad\left(\boldsymbol{\theta}^{N},\mathbf{u}^{M}\right)\in\mathcal{U}}}}{\sum}p\left(\mathbf{\boldsymbol{\varphi}}^{N}\right)f\left(d_{1},0;\gamma,Q,M\right)\left(\sum_{\mathbf{v}^{M}\in\mathbb{F}_{Q}^{M}}p\left(\mathbf{v}^{M}\right)\right)+2\varepsilon\nonumber \\
 & \overset{(c)}{\leq} & \sum_{d_{1}=1}^{\left\lfloor \alpha N\right\rfloor }\underset{^{_{\mathbf{\boldsymbol{\varphi}}^{N}\in\mathbb{F}_{Q}^{N}:\left\Vert \mathbf{\boldsymbol{\varphi}}^{N}-\boldsymbol{\theta}^{N}\right\Vert _{0}=d_{1}}^{\qquad\left(\boldsymbol{\theta}^{N},\mathbf{u}^{M}\right)\in\mathcal{U}}}}{\sum}p\left(\mathbf{\boldsymbol{\varphi}}^{N}\right)f\left(1,0;\gamma,Q,M\right)\nonumber \\
 &  & +\sum_{d_{1}=\left\lceil \alpha N\right\rceil }^{N}\underset{^{_{\mathbf{\boldsymbol{\varphi}}^{N}\in\mathbb{F}_{Q}^{N}:\left\Vert \mathbf{\boldsymbol{\varphi}}^{N}-\boldsymbol{\theta}^{N}\right\Vert _{0}=d_{1}}^{\qquad\left(\boldsymbol{\theta}^{N},\mathbf{u}^{M}\right)\in\mathcal{U}}}}{\sum}p\left(\mathbf{\boldsymbol{\varphi}}^{N}\right)f\left(\left\lceil \alpha N\right\rceil ,0;\gamma,Q,M\right)+2\varepsilon,\label{eq:PeNCUpp5}
\end{eqnarray}
where $(a)$ is by the classification of $\mathbf{\boldsymbol{\varphi}}^{N}$
and $\mathbf{v}^{M}$ according to the $\ell_{0}$ norm of their difference
with $\boldsymbol{\theta}^{N}$ and $\mathbf{u}^{M}$ respectively
and $(b)$ is obtained using the bound (\ref{eq:fd1d2Up}) and using
$\sum_{\mathbf{v}^{M}\in\mathbb{F}_{Q}^{M}}p\left(\mathbf{v}^{M}\right)=1$.
The splitting in $(c)$ permits $f\left(d_{1},0;\gamma,Q,M\right)$
to be bounded in different cases, this idea comes from \cite{rankmin}
and is also meaningful here. The parameter $\alpha$ is a positive
real number with $0<\alpha<0.5$. The way to choose $\alpha$ is discussed
in Section~\ref{sec:SufficientConditionNC-2}. The two terms in (\ref{eq:PeNCUpp5}),
denoted by $P_{\mathcal{U}_{1}}\left(\alpha\right)$ and $P_{\mathcal{U}_{2}}\left(\alpha\right)$,
need to be considered separately. For the first term $P_{\mathcal{U}_{1}}\left(\alpha\right)$,
we have
\begin{eqnarray}
P_{\mathcal{U}_{1}}\left(\alpha\right) & = & f\left(1,0;\gamma,Q,M\right)\sum_{d_{1}=1}^{\left\lfloor \alpha N\right\rfloor }\underset{^{_{\mathbf{\boldsymbol{\varphi}}^{N}\in\mathbb{F}_{Q}^{N}:\left\Vert \mathbf{\boldsymbol{\varphi}}^{N}-\boldsymbol{\theta}^{N}\right\Vert _{0}=d_{1}}^{\qquad\left(\boldsymbol{\theta}^{N},\mathbf{u}^{M}\right)\in\mathcal{U}}}}{\sum}p\left(\mathbf{\boldsymbol{\varphi}}^{N}\right)\nonumber \\
 & \overset{(a)}{=} & \left(1-\gamma\right)^{M}\sum_{\mathbf{u}^{M}\in\mathcal{A}_{\left[\mathsf{u}\right]\varepsilon}^{M}}\sum_{d_{1}=1}^{\left\lfloor \alpha N\right\rfloor }\sum_{\mathbf{\boldsymbol{\varphi}}^{N}\in\mathbb{F}_{Q}^{N}}p\left(\mathbf{\boldsymbol{\varphi}}^{N}\right)\underset{^{_{\left\Vert \boldsymbol{\theta}^{N}-\mathbf{\boldsymbol{\varphi}}^{N}\right\Vert _{0}=d_{1}}^{\quad\boldsymbol{\theta}^{N}\in\mathcal{A}_{\left[\mathsf{\Theta}\right]\varepsilon}^{N}:}}}{\sum}1\nonumber \\
 & \overset{(b)}{\leq} & \left(1-\gamma\right)^{M}\sum_{\mathbf{u}^{M}\in\mathcal{A}_{\left[\mathsf{u}\right]\varepsilon}^{M}}\sum_{d_{1}=1}^{\left\lfloor \alpha N\right\rfloor }\sum_{\mathbf{\boldsymbol{\varphi}}^{N}\in\mathbb{F}_{Q}^{N}}p\left(\mathbf{\boldsymbol{\varphi}}^{N}\right)\left|\left\{ \boldsymbol{\theta}^{N}\in\mathbb{F}_{Q}^{N}:\left\Vert \boldsymbol{\theta}^{N}-\mathbf{\boldsymbol{\varphi}}^{N}\right\Vert _{0}=d_{1}\right\} \right|\nonumber \\
 & \overset{(c)}{\leq} & \left(1-\gamma\right)^{M}\sum_{\mathbf{u}^{M}\in\mathcal{A}_{\left[\mathsf{u}\right]\varepsilon}^{M}}\sum_{d_{1}=1}^{\left\lfloor \alpha N\right\rfloor }2^{NH_{2}\left(\frac{d_{1}}{N}\right)}\left(Q-1\right)^{d_{1}}\nonumber \\
 & \overset{(d)}{\leq} & \left(1-\gamma\right)^{M}\cdot\left|\mathcal{A}_{\left[\mathsf{u}\right]\varepsilon}^{M}\right|\cdot\alpha N\cdot2^{NH_{2}\left(\alpha\right)}\left(Q-1\right)^{\alpha N}\nonumber \\
 & \overset{(e)}{\leq} & 2^{-N\left(-\frac{M}{N}\left(H\left(p_{\mathsf{u}}\right)+\log\left(1-\gamma\right)+\varepsilon\right)-H_{2}\left(\alpha\right)-\alpha\log\left(Q-1\right)-\frac{\log\left(\alpha N\right)}{N}\right)}=P_{1}\left(\alpha\right)\label{eq:PeNCUpp6}
\end{eqnarray}
where $(a)$ is by changing the order of summation and $(b)$ is obtained
considering all $\boldsymbol{\theta}^{N}\in\mathbb{F}_{Q}^{N}$ and
not only typical sequences. The bound $(c)$ is obtained noticing
that 
\begin{eqnarray}
\left|\left\{ \boldsymbol{\theta}^{N}\in\mathbb{F}_{Q}^{N}:\left\Vert \boldsymbol{\theta}^{N}-\mathbf{\boldsymbol{\varphi}}^{N}\right\Vert _{0}=d_{1}\right\} \right| & = & \left(_{d_{1}}^{N}\right)\left(Q-1\right)^{d_{1}}\nonumber \\
 & \leq & 2^{NH_{2}\left(\frac{d_{1}}{N}\right)}\left(Q-1\right)^{d_{1}},\label{eq:PeNCUpp7}
\end{eqnarray}
where $H_{2}\left(p\right)$ denotes the entropy of a Bernoulli-$p$
source and $\sum_{\mathbf{\boldsymbol{\varphi}}^{N}\in\mathbb{F}_{Q}^{N}}p\left(\mathbf{\boldsymbol{\varphi}}^{N}\right)=1$;
$(d)$ is because of the monotonicity of the function $H_{2}\left(\frac{d_{1}}{N}\right)$,
which is increasing in $d_{1}$ as $d_{1}\leq\left\lfloor \alpha N\right\rfloor <N/2$;
$\left(e\right)$ comes from \cite[Theorem 3.1.2]{EIT}, the upper
bound of the size of $\mathcal{A}_{\left[\mathsf{u}\right]\varepsilon}^{M}$,
\emph{i.e.}, 
\begin{equation}
\left|\mathcal{A}_{\left[\mathsf{u}\right]\varepsilon}^{M}\right|\leq2^{M\left(H\left(p_{\mathsf{u}}\right)+\varepsilon\right)},\label{eq:AEPCarU}
\end{equation}
for $M>M_{\varepsilon}$. Similarly, for $N>N_{\varepsilon}$, one
has 
\begin{equation}
\left|\mathcal{A}_{\left[\mathsf{\Theta}\right]\varepsilon}^{N}\right|\leq2^{N\left(\mathcal{H}\left(\mathsf{\mathsf{\Theta}}\right)+\varepsilon\right)}.\label{eq:AEPCarTheta}
\end{equation}
Now we turn to $P_{\mathcal{U}_{2}}\left(\alpha\right)$, 
\begin{eqnarray}
P_{\mathcal{U}_{2}}\left(\alpha\right) & = & \sum_{d_{1}=\left\lceil \alpha
N\right\rceil
}^{N}\underset{^{_{\mathbf{\boldsymbol{\varphi}}^{N}\in\mathbb{F}_{Q}^{N}:\left\Vert
\mathbf{\boldsymbol{\varphi}}^{N}-\boldsymbol{\theta}^{N}\right\Vert
_{0}=d_{1}}^{\quad\,\,\,\,\,\,\,\,\left(\boldsymbol{\theta}^{N},\mathbf{u}^{M}\right)\in\mathcal{U}}}}{\sum}p\left(\mathbf{\boldsymbol{\varphi}}^{N}\right)f\left(\left\lceil \alpha N\right\rceil ,0;\gamma,Q,M\right)\nonumber \\
 & \overset{(a)}{\leq} & \sum_{\left(\boldsymbol{\theta}^{N},\mathbf{u}^{M}\right)\in\mathcal{U}}\sum_{\mathbf{\boldsymbol{\varphi}}^{N}\in\mathbb{F}_{Q}^{N}}p\left(\mathbf{\boldsymbol{\varphi}}^{N}\right)f\left(\left\lceil \alpha N\right\rceil ,0;\gamma,Q,M\right)\nonumber \\
 & = & \left|\mathcal{A}_{\left[\mathsf{\Theta}\right]\varepsilon}^{N}\right|\cdot\left|\mathcal{A}_{\left[\mathsf{u}\right]\varepsilon}^{M}\right|\cdot\left(Q^{-1}+\left(1-\frac{\gamma}{1-Q^{-1}}\right)^{\left\lceil \alpha N\right\rceil }\left(1-Q^{-1}\right)\right)^{M}\nonumber \\
 & \overset{(b)}{\leq} & 2^{-N\left(-\mathcal{H}\left(\mathsf{\mathsf{\Theta}}\right)-\frac{M}{N}\left(H\left(p_{\mathsf{u}}\right)+\log\left(Q^{-1}+\left(1-\frac{\gamma}{1-Q^{-1}}\right)^{\left\lceil \alpha N\right\rceil }\left(1-Q^{-1}\right)\right)+\varepsilon\right)-\varepsilon\right)}=P_{2}\left(\alpha\right),\label{eq:PeNCUpp8}
\end{eqnarray}
where $(a)$ is by ignoring the constraint that $\left\Vert \mathbf{\boldsymbol{\varphi}}^{N}-\boldsymbol{\theta}^{N}\right\Vert _{0}=d_{1}$,
and $(b)$ is by the upper bounds of $\left|\mathcal{A}_{\left[\mathsf{\Theta}\right]\varepsilon}^{N}\right|$
and $\left|\mathcal{A}_{\left[\mathsf{u}\right]\varepsilon}^{M}\right|$,
as before. Equations (\ref{eq:PeNCUpp5}), (\ref{eq:PeNCUpp6}), and
(\ref{eq:PeNCUpp8}) complete the proof.
\end{proof}

\subsection{Sufficient Condition}

\label{sec:SufficientConditionNC-2}

In this section, sufficient conditions for the WN and NC cases are
derived to get a vanishing upper bound of error probability.
\begin{prop}[Sufficient condition, WN and NC cases]
\noindent \label{Prop:SufficientCond NC}Assume the absence of sensing
noise and consider a sensing matrix with sparsity factor $\gamma$.
For some $\delta\in\mathbb{R}^{+}$ (which may be taken arbitrary close
to zero), there exists small positive real numbers $\varepsilon$,
$\xi$, and integers $N_{\delta}$, $M_{\varepsilon}$ such that $\forall N>N_{\delta}$
and $M>M_{\varepsilon}$, if the following conditions hold \end{prop}
\begin{itemize}
\item \noindent \textit{the communication noise is not uniformly distributed,
i.e., 
\begin{equation}
H\left(p_{\mathsf{u}}\right)<\log Q -\xi,\label{eq:SufficientNoise}
\end{equation}
}
\item \noindent \textit{the sparsity factor is lower bounded 
\begin{equation}
\gamma>1-2^{-H\left(p_{\mathsf{u}}\right)-\varepsilon},\label{eq:SufficientGamma}
\end{equation}
}
\item \noindent \textit{the compression ratio $M/N$ satisfies 
\begin{equation}
\frac{M}{N}>\frac{\mathcal{H}\left(\mathsf{\mathsf{\Theta}}\right)+\varepsilon}{\log Q-H\left(p_{\mathsf{u}}\right)-\xi},\label{eq:SufficientNC Ratio}
\end{equation}
}
\end{itemize}
\textit{then one has $P_{\textrm{e}}\leq\delta$ using MAP decoding.
As $N\rightarrow\infty$ and $M\rightarrow\infty$, $\varepsilon$
and $\xi$ can be chosen arbitrary close to zero. }
\begin{proof}
Both $P_{1}\left(\alpha\right)$ and $P_{2}\left(\alpha\right)$ need
to be vanishing for increasing $N$ and $M$. The exponent of each
term is considered respectively. Define, from (\ref{eq:PeNCUpp6}),
\begin{equation}
E_{1}^{\textrm{NC}}=-\frac{M}{N}\left(H\left(p_{\mathsf{u}}\right)+\log\left(1-\gamma\right)+\varepsilon\right)-H_{2}\left(\alpha\right)-\alpha\log\left(Q-1\right)-\frac{\log\left(\alpha N\right)}{N}.\label{eq:E1NC}
\end{equation}
Then $\lim_{N\rightarrow\infty}2^{-NE_{1}^{\textrm{NC}}}=0$ if $E_{1}^{\textrm{NC}}>0$.
Thus, if $E_{1}^{\textrm{NC}}>0$, for any $\tau_{1}\in\mathbb{R}^{+}$
arbitrarily small, $\exists N_{\tau_{1}}$ such that $\forall N>N_{\tau_{1}}$,
one has $P_{1}\left(\alpha\right)<\tau_{1}$.

Notice that if $H\left(p_{\mathsf{u}}\right)+\log\left(1-\gamma\right)+\varepsilon\geq0$,
$E_{1}^{\textrm{NC}}$ is negative, thus one should first have 
\begin{equation}
H\left(p_{\mathsf{u}}\right)+\log\left(1-\gamma\right)+\varepsilon<0,\label{eq:SuffNC1}
\end{equation}
leading to (\ref{eq:SufficientGamma}). With this condition, $E_{1}^{\textrm{NC}}>0$
leads to 
\begin{equation}
\frac{M}{N}>\frac{H_{2}\left(\alpha\right)+\alpha\log\left(Q-1\right)+\frac{\log\left(\alpha N\right)}{N}}{\log\frac{1}{1-\gamma}-H\left(p_{\mathsf{u}}\right)-\varepsilon}.\label{eq:SuffNC2}
\end{equation}
Similarly, define from (\ref{eq:PeNCUpp8})
\begin{equation}
E_{2}^{\textrm{NC}}=-\mathcal{H}\left(\mathsf{\mathsf{\Theta}}\right)-\frac{M}{N}\left(H\left(p_{\mathsf{u}}\right)+\log\left(Q^{-1}+\left(1-\frac{\gamma}{1-Q^{-1}}\right)^{\left\lceil \alpha N\right\rceil }\left(1-Q^{-1}\right)\right)+\varepsilon\right)-\varepsilon.\label{eq:E2NC}
\end{equation}
Again, if $E_{2}^{\textrm{NC}}>0$, for any $\tau_{2}\in\mathbb{R}^{+}$
arbitrarily small, $\exists N_{\tau_{2}}\in\mathbb{N}^{+}$ such that
$\forall N>N_{\tau_{2}}$, one has $P_{2}\left(\alpha\right)<\tau_{2}$.
Since $0<\gamma\leq1-Q^{-1}$, one gets $0\leq1-\frac{\gamma}{1-Q^{-1}}<1$
and 
\begin{equation}
\lim_{N\rightarrow\infty}\left(1-\frac{\gamma}{1-Q^{-1}}\right)^{\left\lceil \alpha N\right\rceil }=0.\label{eq:SuffNC3}
\end{equation}
Thus for $\sigma\in\mathbb{R}^{+}$ arbitrarily small, there exists
an $N_{\sigma}$ such that for $\forall N>N_{\sigma}$, 
\begin{equation}
\left(1-\frac{\gamma}{1-Q^{-1}}\right)^{\left\lceil \alpha N\right\rceil }\left(1-Q^{-1}\right)<\sigma Q^{-1}.\label{eq:SuffNC4}
\end{equation}
Hence $E_{2}^{\textrm{NC}}$ in (\ref{eq:E2NC}) can be lower bounded
by 
\begin{equation}
E_{2}^{\textrm{NC}}>-\mathcal{H}\left(\mathsf{\mathsf{\Theta}}\right)-\frac{M}{N}\left(H\left(p_{\mathsf{u}}\right)+\log\left(Q^{-1}+\sigma Q^{-1}\right)+\varepsilon\right)-\varepsilon,\label{eq:SuffNC5}
\end{equation}
for $N>N_{\sigma}$. If this lower bound is positive, then $E_{2}^{\textrm{NC}}$
is positive. Again, if $H\left(p_{\mathsf{u}}\right)-\log
Q+\log(1+\sigma)+\varepsilon\leq 0$, one
obtains a negative lower bound for $E_{2}^{\textrm{NC}}$ from (\ref{eq:SuffNC5}).
Thus, one deduces (\ref{eq:SufficientNoise}) in
Proposition~\ref{Prop:SufficientCond NC}, with
\begin{equation}
\xi=\log\left(1+\sigma\right)+\varepsilon.\label{eq:Sigma}
\end{equation}
From (\ref{eq:SufficientNoise}), to get a positive lower bound
for (\ref{eq:SuffNC5}), one should have 
\begin{equation}
\frac{M}{N}>\frac{\mathcal{H}\left(\mathsf{\mathsf{\Theta}}\right)+\varepsilon}{\log Q-H\left(p_{\mathsf{u}}\right)-\log\left(1+\sigma\right)-\varepsilon}.\label{eq:SuffNC6}
\end{equation}
From \eqref{eq:SuffNC6} and \eqref{eq:Sigma}, 
with $\xi\rightarrow0$ as $N\rightarrow\infty$, one gets (\ref{eq:SufficientNC Ratio})
in Proposition~\ref{Prop:SufficientCond NC}.

From (\ref{eq:SuffNC2}) and (\ref{eq:SuffNC6}), one obtains 
\begin{equation}
\frac{M}{N}>\max\left\{ \frac{H_{2}\left(\alpha\right)+\alpha\log\left(Q-1\right)+\frac{\log\left(\alpha N\right)}{N}}{\log\frac{1}{1-\gamma}-H\left(p_{\mathsf{u}}\right)-\varepsilon},\:\frac{\mathcal{H}\left(\mathsf{\mathsf{\Theta}}\right)+\varepsilon}{\log Q-H\left(p_{\mathsf{u}}\right)-\xi}\right\} .\label{eq:SuffNC7}
\end{equation}
The value of $\alpha$ should be chosen such that the lower bound
(\ref{eq:SuffNC7}) on $M/N$ is minimum. One may compare (\ref{eq:SuffNC7})
with the necessary condition (\ref{eq:necessary3}). The second term
of (\ref{eq:SuffNC7}) is similar to (\ref{eq:necessary3}), since
both $\xi$ and $\varepsilon$ can be made arbitrarily close to $0$
as $N\rightarrow\infty$. The best value for $\alpha$ has thus to
be such that 
\begin{equation}
\frac{H_{2}\left(\alpha\right)+\alpha\log\left(Q-1\right)+\frac{\log\left(\alpha N\right)}{N}}{\log\frac{1}{1-\gamma}-H\left(p_{\mathsf{u}}\right)-\varepsilon}\leq\frac{\mathcal{H}\left(\mathsf{\mathsf{\Theta}}\right)+\varepsilon}{\log Q-H\left(p_{\mathsf{u}}\right)-\xi}.\label{eq:SuffNC8}
\end{equation}
The function $H_{2}\left(\alpha\right)+\alpha\log\left(Q-1\right)$
is increasing when $\alpha\in\left]0,0.5\right[$ and tends to $0$
as $\alpha\rightarrow0$. The term $\log\left(\alpha N\right)/N$
is also negligible for $N$ large. Thus, there always exists some
$\alpha$ satisfying (\ref{eq:SuffNC8}). Since the speed of convergence
of $\xi$ is affected by $\alpha$, we choose the largest $\alpha$
that satisfies (\ref{eq:SuffNC8}). Finally, the sufficient condition
(\ref{eq:SufficientNC Ratio}) is obtained for $M/N$.

From (\ref{eq:PeUpperNC}), one may conclude that 
\begin{equation}
P_{\textrm{e}}\leq\tau_{1}+\tau_{2}+2\varepsilon.\label{eq:SuffNC9}
\end{equation}
To ensure $P_{\textrm{e}}<\delta$, we should choose $\tau_{1}$,
$\tau_{2}$, and $\varepsilon$ to satisfy $\tau_{1}+\tau_{2}+2\varepsilon<\delta$.
Then a proper value of $\sigma$, which depends on $\tau_{2}$ and
$\varepsilon$, can be chosen. At last, $\xi$ is obtained from (\ref{eq:Sigma}).
With these well determined parameters, if all the three conditions
in Proposition~\ref{Prop:SufficientCond NC} hold, there exists integers
$N_{\varepsilon}$, $N_{\tau_{1}}$, $N_{\tau_{2}}$, and $N_{\sigma}$,
such that for any 
\begin{equation}
N>N_{\delta}=\max\left\{ N_{\varepsilon},N_{\tau_{1}},N_{\tau_{2}}N_{\sigma}\right\} ,
\end{equation}
and $M>M_{\varepsilon}$, one has $P_{\textrm{e}}<\delta$.
\end{proof}

\subsection{Discussion and Numerical Results}

\label{sec:SufficientConditionNC-3}

In \cite[Eq. (24)]{CS}, considering a sparse and iid source, a uniformly
distributed random matrix $\boldsymbol{\mathsf{A}}$, and the minimum
empirical entropy decoder, the following error exponent in the case
NC is obtained 
\begin{equation}
E_{0}^{\textrm{NC}}=\min_{p,q}D\left(p\parallel p_{\mathsf{\Theta}}\right)+\frac{M}{N}D\left(q\parallel p_{\mathsf{u}}\right)+\left|\frac{M}{N}\log Q-H\left(p\right)-\frac{M}{N}H\left(q\right)\right|^{+},\label{eq:cs1}
\end{equation}
where $D\left(\cdot\parallel\cdot\right)$ denotes the relative entropy
between two distributions and $\left|\cdot\right|^{+}=\max\left\{ 0,\cdot\right\} $.
In parallel, \cite{RLNC} proposed an approach to prove that the upper
bound for the probability of decoding error $P_{\textrm{e}}$ under
minimum empirical entropy decoding is equal to that of the maximum
$Q$-probability decoder. As discussed in Section~\ref{sec:SufficientConditionNC-1},
in the WN and NC cases, the MAP decoder in the considered context
is equivalent to the maximum $Q$-probability decoder. As a consequence,
(\ref{eq:cs1}) is also the error exponent of the MAP decoder in the
considered context. A proof for (\ref{eq:cs1}) using the method of
types need to do some assumptions on the topology of the considered
sensor network to specify the type of $\boldsymbol{\theta}^{N}$.
For correlated sources, one can extend (\ref{eq:cs1}) considering
Markov model, and use higher-order types, leading to cumbersome derivations.

From (\ref{eq:cs1}), provided that $E_{0}^{\textrm{NC}}>0$, $P_{\textrm{e}}$
tends to 0 as $N$ increases. $E_{0}^{\textrm{NC}}$ cannot be negative
and $E_{0}^{\textrm{NC}}=0$ if and only if 
\begin{equation}
\begin{cases}
D\left(p\parallel p_{\mathsf{\Theta}}\right)=0,\\
D\left(q\parallel p_{\mathsf{u}}\right)=0,\\
\frac{M}{N}\log Q-H\left(p\right)-\frac{M}{N}H\left(q\right)\leq0.
\end{cases}\label{eq:cs2}
\end{equation}
Thus, (\ref{eq:cs2}) implies that $\frac{M}{N}\log Q-H\left(p_{\mathsf{\Theta}}\right)-\frac{M}{N}H\left(p_{\mathsf{u}}\right)\leq0$.
Thus, a necessary and sufficient condition to have $E_{0}^{\textrm{NC}}>0$
is $\frac{M}{N}\log Q-H\left(p_{\mathsf{\Theta}}\right)-\frac{M}{N}H\left(p_{\mathsf{u}}\right)>0$,
which is the same as (\ref{eq:SufficientNC Ratio}) with $\gamma=1-Q^{-1}$
(corresponding to $\boldsymbol{\mathsf{A}}$ uniformly distributed).
The proof using weak typicality leads to the same results (in terms
of sufficient condition for having asymptotically vanishing $P_{\textrm{e}}$)
as the technique in \cite{CS}.

In the noiseless case, since $\gamma$ can be chosen arbitrarily small,
the necessary condition in Proposition~\ref{Prop:CNforNCSCase} and
the sufficient condition in Proposition~\ref{Prop:SufficientCond NC}
asymptotically coincide. This confirms the numerical results obtained
in \cite{seong}. In the NC case, the difference between the two conditions
comes from the constraint linking $\gamma$ and the entropy of the
communication noise. In Section~\ref{sec:NecessaryCondition}, the
structure of $\boldsymbol{\mathsf{A}}$ was not considered and no
condition on $\gamma$ has been obtained. The lower bound on $\gamma$
implies that $\boldsymbol{\mathsf{A}}$ should be dense enough to
fight against the noise. Since the communication noise is iid, for
a given probability of having one entry of $\boldsymbol{\mathsf{u}}^{M}$
non-zero, \emph{i.e.}, $\Pr\left(\mathsf{u}\neq0\right)$, the entropy
$H\left(p_{\mathsf{u}}\right)$ is maximized when $p_{\mathsf{u}}\left(q\right)=\Pr\left(\mathsf{u}\neq0\right)/\left(Q-1\right)$
for any $q\in\mathbb{F}_{Q}\setminus\left\{ 0\right\} $. This corresponds
to the worst noise in terms of compression efficiency.

Figure~\ref{fig:Lower-bound-gamma} represents the lower bound of
$\gamma$ as a function of $\Pr\left(\mathsf{u}\neq0\right)$, ranging
from $10^{-5}$ to $10^{-1}$, for different value of $Q$. There
is almost no requirement on $\gamma$ when $\Pr\left(\mathsf{u}\neq0\right)\leq5\times10^{-4}$.
For a given noise level, a larger size of the finite field needs a
denser sensing matrix. Figure~\ref{fig:Optimum-compression-ratio}
shows the influence of the communication noise on the optimum compression
ratio. The lower bound of $M/N$ is represented as a function of $\mathcal{H}\left(\mathsf{\mathsf{\Theta}}\right)/\log Q$,
for different values of $Q$ and for different values of $\Pr\left(\mathsf{u}\neq0\right)$. 

\begin{figure}[h]
\begin{centering}
\includegraphics[width=0.65\columnwidth]{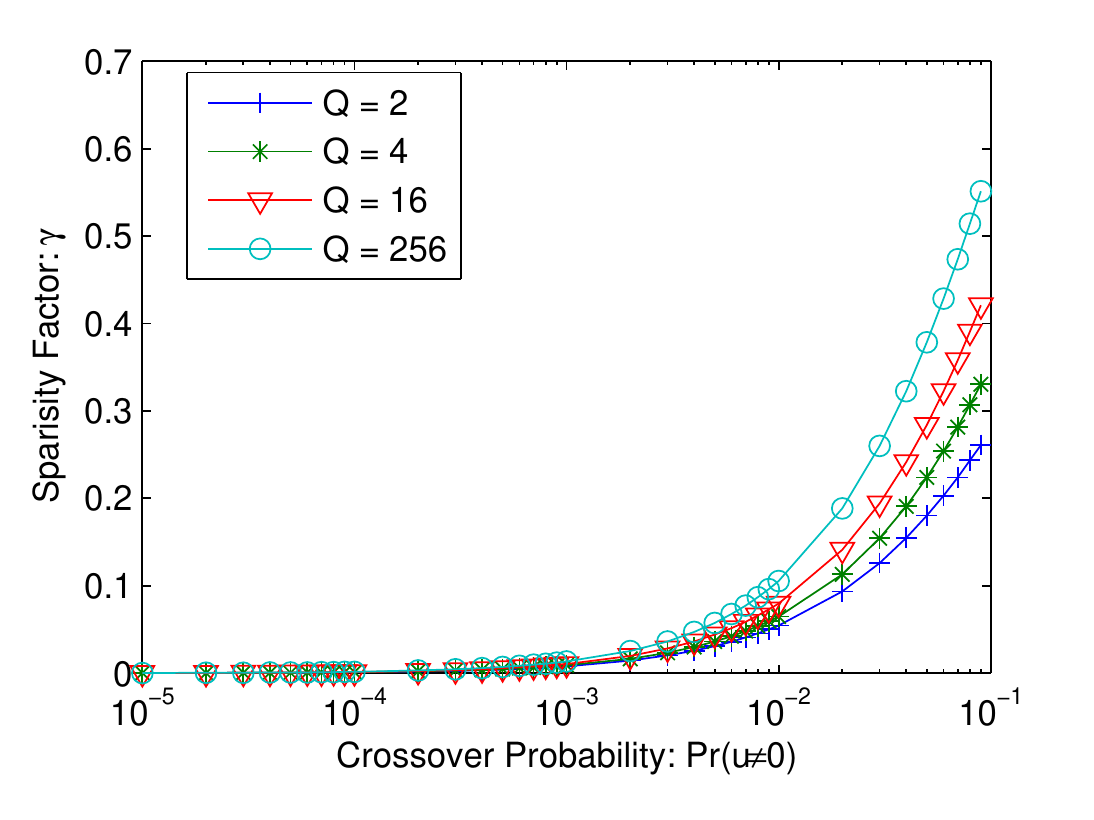}
\par\end{centering}

\caption{Lower bound of $\gamma$ to achieve the optimum compression ratio
for $N\rightarrow\infty$, according to (\ref{eq:SufficientGamma})\label{fig:Lower-bound-gamma} }
\end{figure}
\begin{figure}[th]
\begin{centering}
\includegraphics[width=1\columnwidth]{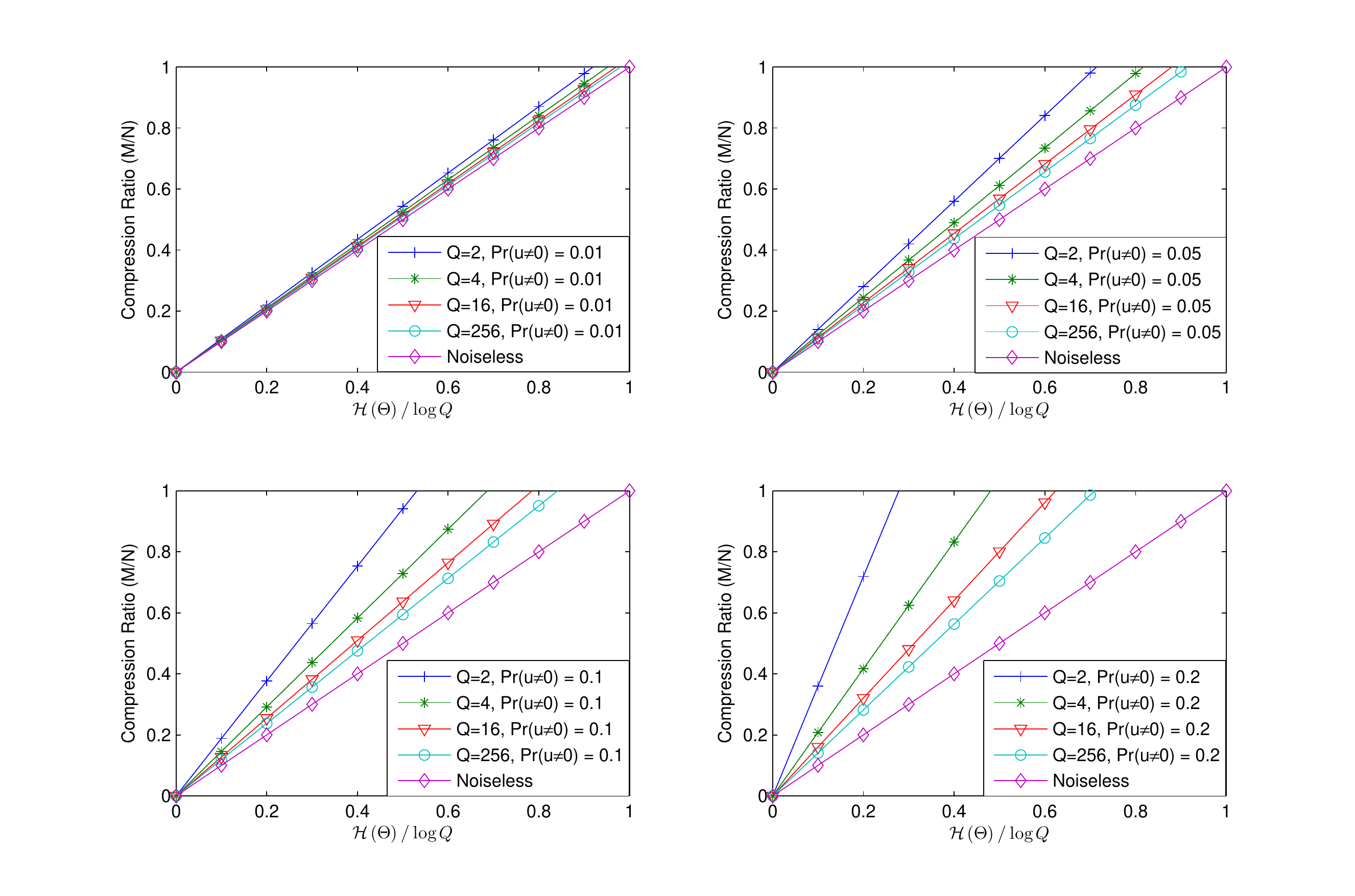}
\par\end{centering}

\caption{Optimum asymptotic achievable compression ratio in function of $\mathcal{H}\left(\mathsf{\mathsf{\Theta}}\right)/\log Q$,
according to (\ref{eq:SufficientNC Ratio}), for a crossover
probability equal to $0.01$, $0.05$, $0.1$, and $0.2$ respectively,
and without noise \label{fig:Optimum-compression-ratio}}
 
\end{figure}

\section{Sufficient Condition in Presence of Sensing Noise}

\label{sec:SufficientConditionNS}

This section performs an achievability study in presence of sensing
noise by considering the conditional pmf $p_{\mathsf{x}\mid\Theta}$.
The communication noise $\mathbf{u}^{M}$ is first neglected to simplify
the problem (NS case). The extension to the NCS case is easily obtained
from the NS case. Assume that $\boldsymbol{\theta}^{N}$ is the true
state vector and that $\mathbf{x}^{N}$ represents the measurements
of the sensors. The sink receives $\mathbf{y}^{M}=\mathbf{A}\mathbf{x}^{N}$.
The \emph{a posteriori} pmf (\ref{eq:MAPEstimatePMF}) can be written
as 
\begin{equation}
p\left(\boldsymbol{\theta}^{N}\mid\mathbf{y}^{M},\mathbf{A}\right)\propto\sum_{\mathbf{z}^{N}\in\mathbb{F}_{Q}^{N}}p\left(\boldsymbol{\theta}^{N}\right)p\left(\mathbf{z}^{N}\mid\boldsymbol{\theta}^{N}\right)1_{\mathbf{y}^{M}=\mathbf{A}\mathbf{z}^{N}}.\label{eq:aposterioriNS}
\end{equation}
In the case of MAP estimation, an error occurs if there exists a vector
$\mathbf{\boldsymbol{\varphi}}^{N}\in\mathbb{F}_{Q}^{N}\setminus\left\{ \boldsymbol{\theta}^{N}\right\} $
such that 
\begin{equation}
\sum_{\mathbf{z}^{N}\in\mathbb{F}_{Q}^{N}}p\left(\boldsymbol{\theta}^{N},\mathbf{z}^{N}\right)1_{\mathbf{y}^{M}=\mathbf{A}\mathbf{z}^{N}}\leq\sum_{\mathbf{z}^{N}\in\mathbb{F}_{Q}^{N}}p\left(\mathbf{\boldsymbol{\varphi}}^{N},\mathbf{z}^{N}\right)1_{\mathbf{y}^{M}=\mathbf{A}\mathbf{z}^{N}}.\label{eq:conditionNS1}
\end{equation}

\noindent $\boldsymbol{\theta}^{N}$ and $\mathbf{x}^{N}$ are considered
as fixed, but unknown. The decoder has knowledge of $\mathbf{A}$
and $\mathbf{y}^{M}=\mathbf{A}\mathbf{x}^{N}$, thus an alternative
way to express (\ref{eq:conditionNS1}) is 
\begin{equation}
\sum_{\mathbf{z}^{N}\in\mathbb{F}_{Q}^{N}}p\left(\boldsymbol{\theta}^{N},\mathbf{z}^{N}\right)1_{\mathbf{A}\mathbf{x}^{N}=\mathbf{A}\mathbf{z}^{N}}\leq\sum_{\mathbf{z}^{N}\in\mathbb{F}_{Q}^{N}}p\left(\mathbf{\boldsymbol{\varphi}}^{N},\mathbf{z}^{N}\right)1_{\mathbf{A}\mathbf{x}^{N}=\mathbf{A}\mathbf{z}^{N}}.\label{eq:conditionNS2}
\end{equation}

\subsection{Achievability Study}

\label{sec:SufficientConditionNS-1}

We begin with the extension of the basic weakly typical set as introduced
in Section~\ref{sec:SufficientConditionNC-1}. For any $\varepsilon>0$
and $N\in\mathbb{N}^{+}$, based on $\mathcal{A}_{\left[\mathsf{\Theta}\right]\varepsilon}^{N}$
for $\boldsymbol{\theta}^{N}$, one defines the weakly conditional
typical set $\mathcal{A}_{\left[\mathsf{x}\mid\mathsf{\Theta}\right]\varepsilon}^{N}\left(\boldsymbol{\theta}^{N}\right)$
for $\mathbf{x}^{N}$, which is conditionally distributed with respect
to $p_{\mathsf{x}\mid\mathsf{\Theta}}$, with $\boldsymbol{\theta}^{N}\in\mathcal{A}_{\left[\mathsf{\Theta}\right]\varepsilon}^{N}$,
{
\begin{equation}
\mathcal{A}_{\left[\mathsf{x}\mid\mathsf{\Theta}\right]\varepsilon}^{N}\left(\boldsymbol{\theta}^{N}\right)=\left\{ \mathbf{x}^{N}\in\mathbb{F}_{Q}^{N}\mbox{ such that }\left|-\frac{1}{N}\log p\left(\mathbf{x}^{N}\mid\boldsymbol{\theta}^{N}\right)-\mathcal{H}\left(\mathsf{x}\mid\mathsf{\Theta}\right)\right|\leq\varepsilon\right\} .\label{eq:ConditionalTypical}
\end{equation}
}Since $\mathcal{H}\left(\mathsf{\mathsf{\Theta}},\mathsf{x}\right)=\mathcal{H}\left(\mathsf{\mathsf{\Theta}}\right)+\mathcal{H}\left(\mathsf{x}\mid\mathsf{\Theta}\right)$,
if $\boldsymbol{\theta}^{N}\in\mathcal{A}_{\left[\mathsf{\Theta}\right]\varepsilon}^{N}$
and $\mathbf{x}^{N}\in\mathcal{A}_{\left[\mathsf{x}\mid\mathsf{\Theta}\right]\varepsilon}^{N}\left(\boldsymbol{\theta}^{N}\right)$,
then $\left(\boldsymbol{\theta}^{N},\mathbf{x}^{N}\right)\in\mathcal{A}_{\left[\mathsf{\Theta},\mathsf{x}\right]2\varepsilon}^{N}$
by consistency, where $\mathcal{A}_{\left[\mathsf{\Theta},\mathsf{x}\right]2\varepsilon}^{N}$
denotes the weakly joint typical set, \emph{i.e.}, the set of pairs
$\left(\boldsymbol{\theta}^{N},\mathbf{x}^{N}\right)\in\mathbb{F}_{Q}^{N}\times\mathbb{F}_{Q}^{N}$
such that 
\begin{equation}
\left|-\frac{1}{N}\log p\left(\boldsymbol{\theta}^{N},\mathbf{x}^{N}\right)-\mathcal{H}\left(\mathsf{\mathsf{\Theta}},\mathsf{x}\right)\right|\leq2\varepsilon.\label{eq:JointTypical}
\end{equation}
For any \textit{$\varepsilon>0$} there exist an $N_{\varepsilon}$
such that for all $N\geq N_{\varepsilon}$ and for any $\boldsymbol{\theta}^{N}\in\mathcal{A}_{\left[\mathsf{\Theta}\right]\varepsilon}^{N}$,
one has $\Pr\left\{ \boldsymbol{\mathsf{x}}^{N}\in\mathcal{A}_{\left[\mathsf{x}\mid\mathsf{\Theta}\right]\varepsilon}^{N}\left(\boldsymbol{\theta}^{N}\right)\right\} \geq1-\varepsilon$
and $\Pr\left\{ \left(\boldsymbol{\mathsf{\Theta}}^{N},\boldsymbol{\mathsf{x}}^{N}\right)\in\mathcal{A}_{\left[\mathsf{\Theta},\mathsf{x}\right]2\varepsilon}^{N}\right\} \geq1-2\varepsilon$.
The cardinality of the set $\mathcal{A}_{\left[\mathsf{\Theta},\mathsf{x}\right]2\varepsilon}^{N}$
satisfies
\begin{equation}
\left|\mathcal{A}_{\left[\mathsf{\Theta},\mathsf{x}\right]2\varepsilon}^{N}\right|\leq2^{N\left(\mathcal{H}\left(\mathsf{\mathsf{\Theta}},\mathsf{x}\right)+2\varepsilon\right)}.\label{eq:AEPJoint}
\end{equation}
One may have $\varepsilon$ arbitrary close to zero as $N\rightarrow\infty$.

Considering $\mathcal{A}_{\left[\mathsf{\Theta},\mathsf{x}\right]2\varepsilon}^{N}$,
the estimation error probability is bounded by
\begin{eqnarray}
P_{\textrm{e}} & \leq & \sum_{\left(\boldsymbol{\theta}^{N},\mathbf{x}^{N}\right)\in\mathcal{A}_{\left[\mathsf{\Theta},\mathsf{x}\right]2\varepsilon}^{N}}p\left(\boldsymbol{\theta}^{N},\mathbf{x}^{N}\right)\Pr\left\{ \textrm{error}\mid\boldsymbol{\theta}^{N},\mathbf{x}^{N}\right\} +\sum_{\left(\boldsymbol{\theta}^{N},\mathbf{x}^{N}\right)\notin\mathcal{A}_{\left[\mathsf{\Theta},\mathsf{x}\right]2\varepsilon}^{N}}p\left(\boldsymbol{\theta}^{N},\mathbf{x}^{N}\right)\nonumber \\
 & \leq & \sum_{\left(\boldsymbol{\theta}^{N},\mathbf{x}^{N}\right)\in\mathcal{A}_{\left[\mathsf{\Theta},\mathsf{x}\right]2\varepsilon}^{N}}p\left(\boldsymbol{\theta}^{N},\mathbf{x}^{N}\right)\cdot\Pr\left\{ \textrm{error}\mid\boldsymbol{\theta}^{N},\mathbf{x}^{N}\right\} +2\varepsilon,\label{eq:PeNS}
\end{eqnarray}
Errors appear mainly due to a bad sensing matrix. Averaging over all
$\mathbf{A}\in\mathbb{F}_{Q}^{M\times N}$, (\ref{eq:PeNS}) becomes
\begin{equation}
P_{\textrm{e}}\leq\sum_{\mathbf{A}\in\mathbb{F}_{Q}^{M\times N}}p\left(\mathbf{A}\right)\sum_{\boldsymbol{\theta}^{N}\in\mathcal{A}_{\left[\mathsf{\Theta}\right]\varepsilon}^{N}}\sum_{\mathbf{x}^{N}\in\mathcal{A}_{\left[\mathsf{x}\mid\mathsf{\Theta}\right]\varepsilon}^{N}\left(\boldsymbol{\theta}^{N}\right)}p\left(\boldsymbol{\theta}^{N},\mathbf{x}^{N}\right)\Pr\left\{ \textrm{error}\mid\boldsymbol{\theta}^{N},\mathbf{x}^{N},\mathbf{A}\right\} +2\varepsilon,\label{eq:PeNS1}
\end{equation}
where $p\left(\mathbf{A}\right)=\Pr\left\{ \boldsymbol{\mathsf{A}}=\mathbf{A}\right\} $.
$\Pr\left\{ \textrm{error}\mid\boldsymbol{\theta}^{N},\mathbf{x}^{N},\mathbf{A}\right\} $
can be written as 
\begin{equation}
\Pr\left\{ \textrm{error}\mid\boldsymbol{\theta}^{N},\mathbf{x}^{N},\mathbf{A}\right\} =\begin{cases}
1 & \mbox{ if }\exists\,\mathbf{\boldsymbol{\varphi}}^{N}\in\mathbb{F}_{Q}^{N}\setminus\left\{ \boldsymbol{\theta}^{N}\right\} \textrm{ s.t. \eqref{eq:conditionNS2} holds,}\\
0 & \mbox{ if }\forall\,\mathbf{\boldsymbol{\varphi}}^{N}\in\mathbb{F}_{Q}^{N}\setminus\left\{ \boldsymbol{\theta}^{N}\right\} ,\textrm{ \eqref{eq:conditionNS2} does not hold.}
\end{cases}\label{eq:PeNCCond}
\end{equation}
Using again the idea of Lemma~1,
the conditional
error probability is bounded by 
\begin{equation}
\Pr\left\{ \textrm{error}\mid\boldsymbol{\theta}^{N},\mathbf{x}^{N},\mathbf{A}\right\} \leq\sum_{\mathbf{\boldsymbol{\varphi}}^{N}\in\mathbb{F}_{Q}^{N}\setminus\left\{ \boldsymbol{\theta}^{N}\right\} }\frac{\sum_{\mathbf{z}_{1}^{N}\in\mathbb{F}_{Q}^{N}}p\left(\mathbf{\boldsymbol{\varphi}}^{N},\mathbf{z}_{1}^{N}\right)1_{\mathbf{A}\mathbf{x}^{N}=\mathbf{A}\mathbf{z}_{1}^{N}}}{\sum_{\mathbf{z}_{2}^{N}\in\mathbb{F}_{Q}^{N}}p\left(\boldsymbol{\theta}^{N},\mathbf{z}_{2}^{N}\right)1_{\mathbf{A}\mathbf{x}^{N}=\mathbf{A}\mathbf{z}_{2}^{N}}}.\label{eq:PeNsCond Upper}
\end{equation}
From (\ref{eq:PeNS1}) and (\ref{eq:PeNsCond Upper}), one gets 
\begin{equation}
P_{\textrm{e}}\leq\sum_{_{\boldsymbol{\theta}^{N}\in\mathcal{A}_{\left[\mathsf{\Theta}\right]\varepsilon}^{N}}^{\mathbf{A}\in\mathbb{F}_{Q}^{M\times N}}}p\left(\mathbf{A}\right)\hspace{-4mm}\sum_{\mathbf{x}^{N}\in\mathcal{A}_{\left[\mathsf{x}\mid\mathsf{\Theta}\right]\varepsilon}^{N}\left(\boldsymbol{\theta}^{N}\right)}\hspace{-5mm}p\left(\boldsymbol{\theta}^{N},\mathbf{x}^{N}\right)\hspace{-3mm}\sum_{\mathbf{\boldsymbol{\varphi}}^{N}\in\mathbb{F}_{Q}^{N}\setminus\left\{ \boldsymbol{\theta}^{N}\right\} }\frac{\sum_{\mathbf{z}_{1}^{N}\in\mathbb{F}_{Q}^{N}}p\left(\mathbf{\boldsymbol{\varphi}}^{N},\mathbf{z}_{1}^{N}\right)1_{\mathbf{A}\mathbf{x}^{N}=\mathbf{A}\mathbf{z}_{1}^{N}}}{\sum_{\mathbf{z}_{2}^{N}\in\mathbb{F}_{Q}^{N}}p\left(\boldsymbol{\theta}^{N},\mathbf{z}_{2}^{N}\right)1_{\mathbf{A}\mathbf{x}^{N}=\mathbf{A}\mathbf{z}_{2}^{N}}}+2\varepsilon.\label{eq:PeNS2}
\end{equation}
Now, for some $\boldsymbol{\theta}^{N}\in\mathcal{A}_{\left[\mathsf{\Theta}\right]\varepsilon}^{N}$,
consider the direct image by $\mathbf{A}$ of the conditional typical
set $\mathcal{A}_{\left[\mathsf{x}\mid\mathsf{\Theta}\right]\varepsilon}^{N}\left(\boldsymbol{\theta}^{N}\right)$
\begin{equation}
\mathcal{Y}_{\varepsilon}\left(\mathbf{A},\boldsymbol{\theta}^{N}\right)=\left\{ \mathbf{y}^{M}=\mathbf{A}\mathbf{x}^{N},\textrm{ for all }\;\mathbf{x}^{N}\in\mathcal{A}_{\left[\mathsf{x}\mid\mathsf{\Theta}\right]\varepsilon}^{N}\left(\boldsymbol{\theta}^{N}\right)\right\} .\label{eq:Yset}
\end{equation}
\bigskip{} 
\noindent \textbf{Lemma 3.} \textit{For any arbitrary real-valued
function $h\left(\mathbf{x}^{N}\right)$ with $\mathbf{x}^{N}\in\mathbb{F}_{Q}^{N}$,
one has} 
\begin{equation}
\sum_{\mathbf{x}^{N}\in\mathcal{A}_{\left[\mathsf{x}\mid\mathsf{\Theta}\right]\varepsilon}^{N}\left(\boldsymbol{\theta}^{N}\right)}h\left(\mathbf{x}^{N}\right)=\sum_{\mathbf{y}^{M}\in\mathcal{Y}_{\varepsilon}\left(\mathbf{A},\boldsymbol{\theta}^{N}\right)}\sum_{\mathbf{x}^{N}\in\mathcal{A}_{\left[\mathsf{x}\mid\mathsf{\Theta}\right]\varepsilon}^{N}\left(\boldsymbol{\theta}^{N}\right)}h\left(\mathbf{x}^{N}\right)1_{\mathbf{y}^{M}=\mathbf{A}\mathbf{x}^{N}}.\label{eq:Seperation1}
\end{equation}

\begin{proof}
For a given $\mathbf{y}^{M}\in\mathcal{Y}_{\varepsilon}\left(\mathbf{A},\boldsymbol{\theta}^{N}\right)$,
consider the set 
\begin{equation}
\mathcal{X}_{\varepsilon}\left(\mathbf{y}^{M},\mathbf{A},\boldsymbol{\theta}^{N}\right)=\left\{ \mathbf{x}^{N}\in\mathcal{A}_{\left[\mathsf{x}\mid\mathsf{\Theta}\right]\varepsilon}^{N}\left(\boldsymbol{\theta}^{N}\right)\mbox{ such that }\mathbf{y}^{M}=\mathbf{A}\mathbf{x}^{N}\right\} .\label{eq:Xset}
\end{equation}
Then one has 
\begin{equation}
\mathcal{A}_{\left[\mathsf{x}\mid\mathsf{\Theta}\right]\varepsilon}^{N}\left(\boldsymbol{\theta}^{N}\right)=\bigcup_{\mathbf{y}^{M}\in\mathcal{Y}_{\varepsilon}\left(\mathbf{A},\boldsymbol{\theta}^{N}\right)}\mathcal{X}_{\varepsilon}\left(\mathbf{y}^{M},\mathbf{A},\boldsymbol{\theta}^{N}\right),\label{eq:Seperation2}
\end{equation}
with $\mathcal{X}_{\varepsilon}\left(\mathbf{y}_{i}^{M},\mathbf{A},\boldsymbol{\theta}^{N}\right)\cap\mathcal{X}_{\varepsilon}\left(\mathbf{y}_{j}^{M},\mathbf{A},\boldsymbol{\theta}^{N}\right)=\emptyset$
for any $\mathbf{y}_{i}^{M}\neq\mathbf{y}_{j}^{M},$ since the multiplication
by $\mathbf{A}$ is a surjection from $\mathcal{A}_{\left[\mathsf{x}\mid\mathsf{\Theta}\right]\varepsilon}^{N}\left(\boldsymbol{\theta}^{N}\right)$
to $\mathcal{Y}_{\varepsilon}\left(\mathbf{A},\boldsymbol{\theta}^{N}\right)$.
So any sum over $\mathbf{x}^{N}\in\mathcal{A}_{\left[\mathsf{x}\mid\mathsf{\Theta}\right]\varepsilon}^{N}\left(\boldsymbol{\theta}^{N}\right)$
can be decomposed as
\begin{eqnarray}
\sum_{\mathbf{x}^{N}\in\mathcal{A}_{\left[\mathsf{x}\mid\mathsf{\Theta}\right]\varepsilon}^{N}\left(\boldsymbol{\theta}^{N}\right)}h\left(\mathbf{x}^{N}\right) & = & \sum_{\mathbf{y}^{M}\in\mathcal{Y}_{\varepsilon}\left(\mathbf{A},\boldsymbol{\theta}^{N}\right)}\sum_{\mathbf{x}^{N}\in\mathcal{X}_{\varepsilon}\left(\mathbf{y}^{M},\mathbf{A},\boldsymbol{\theta}^{N}\right)}h\left(\mathbf{x}^{N}\right)\nonumber \\
 & = & \sum_{\mathbf{y}^{M}\in\mathcal{Y}_{\varepsilon}\left(\mathbf{A},\boldsymbol{\theta}^{N}\right)}\sum_{\mathbf{x}^{N}\in\mathcal{A}_{\left[\mathsf{x}\mid\mathsf{\Theta}\right]\varepsilon}^{N}\left(\boldsymbol{\theta}^{N}\right)}h\left(\mathbf{x}^{N}\right)1_{\mathbf{y}^{M}=\mathbf{A}\mathbf{x}^{N}}.\label{eq:Seperation4}
\end{eqnarray}

\end{proof}
Applying (\ref{eq:Seperation1}) to (\ref{eq:PeNS2}), one obtains
\begin{align}
P_{\textrm{e}} & \leq \sum_{_{\boldsymbol{\theta}^{N}\in\mathcal{A}_{\left[\mathsf{\Theta}\right]\varepsilon}^{N}}^{\mathbf{A}\in\mathbb{F}_{Q}^{M\times N}}}p\left(\mathbf{A}\right)\sum_{\mathbf{y}^{M}\in\mathcal{Y}_{\varepsilon}\left(\mathbf{A},\boldsymbol{\theta}^{N}\right)}\sum_{\mathbf{x}^{N}\in\mathcal{A}_{\left[\mathsf{x}\mid\mathsf{\Theta}\right]\varepsilon}^{N}\left(\boldsymbol{\theta}^{N}\right)}p\left(\boldsymbol{\theta}^{N},\mathbf{x}^{N}\right)1_{\mathbf{y}^{M}=\mathbf{A}\mathbf{x}^{N}}\nonumber \\
 &   \cdot\left(\sum_{\mathbf{\boldsymbol{\varphi}}^{N}\in\mathbb{F}_{Q}^{N}\setminus\left\{ \boldsymbol{\theta}^{N}\right\} }\frac{\sum_{\mathbf{z}_{1}^{N}\in\mathbb{F}_{Q}^{N}}p\left(\mathbf{\boldsymbol{\varphi}}^{N},\mathbf{z}_{1}^{N}\right)1_{\mathbf{y}^{M}=\mathbf{A}\mathbf{z}_{1}^{N}}}{\sum_{\mathbf{z}_{2}^{N}\in\mathbb{F}_{Q}^{N}}p\left(\boldsymbol{\theta}^{N},\mathbf{z}_{2}^{N}\right)1_{\mathbf{y}^{M}=\mathbf{A}\mathbf{z}_{2}^{N}}}\right)+2\varepsilon\nonumber \\
 & =  \sum_{_{\boldsymbol{\theta}^{N}\in\mathcal{A}_{\left[\mathsf{\Theta}\right]\varepsilon}^{N}}^{\mathbf{A}\in\mathbb{F}_{Q}^{M\times N}}}p\left(\mathbf{A}\right)\sum_{\mathbf{y}^{M}\in\mathcal{Y}_{\varepsilon}\left(\mathbf{A},\boldsymbol{\theta}^{N}\right)}\left(\sum_{\mathbf{\boldsymbol{\varphi}}^{N}\in\mathbb{F}_{Q}^{N}\setminus\left\{ \boldsymbol{\theta}^{N}\right\} }\sum_{\mathbf{z}_{1}^{N}\in\mathbb{F}_{Q}^{N}}p\left(\mathbf{\boldsymbol{\varphi}}^{N},\mathbf{z}_{1}^{N}\right)\cdot1_{\mathbf{y}^{M}=\mathbf{A}\mathbf{z}_{1}^{N}}\right)\nonumber \\
 &   \cdot\left(\frac{\sum_{\mathbf{x}^{N}\in\mathcal{A}_{\left[\mathsf{x}\mid\mathsf{\Theta}\right]\varepsilon}^{N}\left(\boldsymbol{\theta}^{N}\right)}p\left(\boldsymbol{\theta}^{N},\mathbf{x}^{N}\right)1_{\mathbf{y}^{M}=\mathbf{A}\mathbf{x}^{N}}}{\sum_{\mathbf{z}_{2}^{N}\in\mathbb{F}_{Q}^{N}}p\left(\boldsymbol{\theta}^{N},\mathbf{z}_{2}^{N}\right)1_{\mathbf{y}^{M}=\mathbf{A}\mathbf{z}_{2}^{N}}}\right)+2\varepsilon\nonumber \\
 & \leq  \sum_{_{\boldsymbol{\theta}^{N}\in\mathcal{A}_{\left[\mathsf{\Theta}\right]\varepsilon}^{N}}^{\mathbf{A}\in\mathbb{F}_{Q}^{M\times N}}}p\left(\mathbf{A}\right)\sum_{\mathbf{y}^{M}\in\mathcal{Y}_{\varepsilon}\left(\mathbf{A},\boldsymbol{\theta}^{N}\right)}\left(\sum_{\mathbf{\boldsymbol{\varphi}}^{N}\in\mathbb{F}_{Q}^{N}\setminus\left\{ \boldsymbol{\theta}^{N}\right\} }\sum_{\mathbf{z}_{1}^{N}\in\mathbb{F}_{Q}^{N}}p\left(\mathbf{\boldsymbol{\varphi}}^{N},\mathbf{z}_{1}^{N}\right)1_{\mathbf{y}^{M}=\mathbf{A}\mathbf{z}_{1}^{N}}\right)+2\varepsilon,\label{eq:PeNS3}
\end{align}
since we have 
\begin{equation}
\frac{\sum_{\mathbf{x}^{N}\in\mathcal{A}_{\left[\mathsf{x}\mid\mathsf{\Theta}\right]\varepsilon}^{N}\left(\boldsymbol{\theta}^{N}\right)}p\left(\boldsymbol{\theta}^{N},\mathbf{x}^{N}\right)1_{\mathbf{y}^{M}=\mathbf{A}\mathbf{x}^{N}}}{\sum_{\mathbf{z}_{2}^{N}\in\mathbb{F}_{Q}^{N}}p\left(\boldsymbol{\theta}^{N},\mathbf{z}_{2}^{N}\right)1_{\mathbf{y}^{M}=\mathbf{A}\mathbf{z}_{2}^{N}}}\leq1.\label{eq:PeNS4}
\end{equation}
The bound (\ref{eq:PeNS4}) is tight because for $N$ sufficiently
large, the probability of the non-typical set vanishes. Recall that
$\mathbf{y}^{M}=\mathbf{A}\mathbf{x}^{N}$, even though $\mathbf{x}^{N}$
is not explicit in (\ref{eq:PeNS3}). As a vector $\mathbf{y}^{M}$
may correspond to several $\mathbf{x}^{N}$s, (\ref{eq:PeNS3}) is
further bounded by
\begin{eqnarray}
P_{\textrm{e}} & \leq & \sum_{_{\boldsymbol{\theta}^{N}\in\mathcal{A}_{\left[\mathsf{\Theta}\right]\varepsilon}^{N}}^{\mathbf{A}\in\mathbb{F}_{Q}^{M\times N}}}p\left(\mathbf{A}\right)\sum_{\mathbf{x}^{N}\in\mathcal{A}_{\left[\mathsf{x}\mid\mathsf{\Theta}\right]\varepsilon}^{N}\left(\boldsymbol{\theta}^{N}\right)}\sum_{_{\quad\mathbf{z}_{1}^{N}\in\mathbb{F}_{Q}^{N}}^{\mathbf{\boldsymbol{\varphi}}^{N}\in\mathbb{F}_{Q}^{N}\setminus\left\{ \boldsymbol{\theta}^{N}\right\} }}p\left(\mathbf{\boldsymbol{\varphi}}^{N},\mathbf{z}_{1}^{N}\right)1_{\mathbf{A}\mathbf{x}^{N}=\mathbf{A}\mathbf{z}_{1}^{N}}+2\varepsilon\nonumber \\
 & \leq & \sum_{_{\mathbf{\boldsymbol{\varphi}}^{N}\in\mathbb{F}_{Q}^{N}\setminus\left\{ \boldsymbol{\theta}^{N}\right\} }^{\;\;\;\;\boldsymbol{\theta}^{N}\in\mathcal{A}_{\left[\mathsf{\Theta}\right]\varepsilon}^{N}}}\sum_{_{\;\:\quad\mathbf{z}_{1}^{N}\in\mathbb{F}_{Q}^{N}}^{\mathbf{x}^{N}\in\mathcal{A}_{\left[\mathsf{x}\mid\mathsf{\Theta}\right]\varepsilon}^{N}\left(\boldsymbol{\theta}^{N}\right)}}p\left(\mathbf{\boldsymbol{\varphi}}^{N},\mathbf{z}_{1}^{N}\right)\sum_{\mathbf{A}\in\mathbb{F}_{Q}^{M\times N}}p\left(\mathbf{A}\right)1_{\mathbf{A}\mathbf{x}^{N}=\mathbf{A}\mathbf{z}_{1}^{N}}+2\varepsilon.\label{eq:PeNS5}
\end{eqnarray}
Since 
\begin{equation}
\sum_{\mathbf{A}\in\mathbb{F}_{Q}^{M\times N}}p\left(\mathbf{A}\right)1_{\mathbf{A}\mathbf{x}^{N}=\mathbf{A}\mathbf{z}_{1}^{N}}=\Pr\left\{ \boldsymbol{\mathsf{A}}\mathbf{x}^{N}=\boldsymbol{\mathsf{A}}\mathbf{z}_{1}^{N}\right\} ,\label{eq:PA}
\end{equation}
one gets
\begin{equation}
P_{\textrm{e}}\leq\sum_{_{\mathbf{\boldsymbol{\varphi}}^{N}\in\mathbb{F}_{Q}^{N}\setminus\left\{ \boldsymbol{\theta}^{N}\right\} }^{\;\;\;\;\boldsymbol{\theta}^{N}\in\mathcal{A}_{\left[\mathsf{\Theta}\right]\varepsilon}^{N}}}\sum_{_{\;\:\quad\mathbf{z}_{1}^{N}\in\mathbb{F}_{Q}^{N}}^{\mathbf{x}^{N}\in\mathcal{A}_{\left[\mathsf{x}\mid\mathsf{\Theta}\right]\varepsilon}^{N}\left(\boldsymbol{\theta}^{N}\right)}}p\left(\mathbf{\boldsymbol{\varphi}}^{N},\mathbf{z}_{1}^{N}\right)\Pr\left\{ \boldsymbol{\mathsf{A}}\mathbf{x}^{N}=\boldsymbol{\mathsf{A}}\mathbf{z}_{1}^{N}\right\} +2\varepsilon.\label{eq:PeNS6}
\end{equation}
Suppose that $\left\Vert \mathbf{x}^{N}-\mathbf{z}_{1}^{N}\right\Vert _{0}=d$.
If $d=0$, $\Pr\left\{ \boldsymbol{\mathsf{A}}\mathbf{x}^{N}=\boldsymbol{\mathsf{A}}\mathbf{z}_{1}^{N}\right\} $
equals 1. Otherwise we can apply Lemma~2,
without communication noise, $\Pr\left\{ \boldsymbol{\mathsf{A}}\mathbf{x}^{N}=\boldsymbol{\mathsf{A}}\mathbf{z}_{1}^{N}\right\} =f\left(d,0;\gamma,Q,M\right)$
. Depending on $d$ being zero or not, $P_{\mathcal{A}}$ is split
as follows
\begin{equation}
P_{\textrm{e}}\leq P_{\mathcal{A}_{1}}+P_{\mathcal{A}_{2}}+2\varepsilon,\label{eq:PeNS7}
\end{equation}
where 
\begin{equation}
P_{\mathcal{A}_{1}}=\sum_{_{\mathbf{\boldsymbol{\varphi}}^{N}\in\mathbb{F}_{Q}^{N}\setminus\left\{ \boldsymbol{\theta}^{N}\right\} }^{\;\;\;\;\boldsymbol{\theta}^{N}\in\mathcal{A}_{\left[\mathsf{\Theta}\right]\varepsilon}^{N}}}\sum_{\mathbf{z}_{1}^{N}\in\mathcal{A}_{\left[\mathsf{x}\mid\mathsf{\Theta}\right]\varepsilon}^{N}\left(\boldsymbol{\theta}^{N}\right)}p\left(\mathbf{\boldsymbol{\varphi}}^{N},\mathbf{z}_{1}^{N}\right),\label{eq:PA1}
\end{equation}
and 
\begin{equation}
P_{\mathcal{A}_{2}}=\sum_{_{\mathbf{\boldsymbol{\varphi}}^{N}\in\mathbb{F}_{Q}^{N}\setminus\left\{ \boldsymbol{\theta}^{N}\right\} }^{\;\;\;\;\boldsymbol{\theta}^{N}\in\mathcal{A}_{\left[\mathsf{\Theta}\right]\varepsilon}^{N}}}\sum_{_{\:\:\:\mathbf{z}_{1}^{N}\in\mathbb{F}_{Q}^{N}\setminus\left\{ \mathbf{x}^{N}\right\} }^{\mathbf{x}^{N}\in\mathcal{A}_{\left[\mathsf{x}\mid\mathsf{\Theta}\right]\varepsilon}^{N}\left(\boldsymbol{\theta}^{N}\right)}}p\left(\mathbf{\boldsymbol{\varphi}}^{N},\mathbf{z}_{1}^{N}\right)\Pr\left\{ \boldsymbol{\mathsf{A}}\mathbf{x}^{N}=\boldsymbol{\mathsf{A}}\mathbf{z}_{1}^{N}\right\} .\label{eq:PA2}
\end{equation}

\noindent \textbf{Lemma 4.} \textit{
A sufficient condition for $P_{\mathcal{A}_{1}}\leq2\varepsilon$ 
is that,
for any pair of vectors $(\boldsymbol{\theta}^{N},\mathbf{\boldsymbol{\varphi}}^{N})\in\mathcal{A}_{\left[\mathsf{\Theta}\right]\varepsilon}^{N}\times\mathcal{A}_{\left[\mathsf{\Theta}\right]\varepsilon}^{N}$
such that $\boldsymbol{\theta}^{N}\neq\mathbf{\boldsymbol{\varphi}}^{N}$,
\begin{equation}
\mathcal{A}_{\left[\mathsf{x}\mid\mathsf{\Theta}\right]\varepsilon}^{N}\left(\boldsymbol{\theta}^{N}\right)\cap\mathcal{A}_{\left[\mathsf{x}\mid\mathsf{\Theta}\right]\varepsilon}^{N}\left(\mathbf{\boldsymbol{\varphi}}^{N}\right)=\emptyset.\label{eq:overlapping}
\end{equation}
}
\begin{proof}
Assume that (\ref{eq:overlapping}) is satisfied. Changing the order
of summation, (\ref{eq:PA1}) becomes
\begin{equation}
P_{\mathcal{A}_{1}}=\sum_{\mathbf{\boldsymbol{\varphi}}^{N}\in\mathbb{F}_{Q}^{N}}p\left(\mathbf{\boldsymbol{\varphi}}^{N}\right)\sum_{_{\,\mathbf{z}_{1}^{N}\in\mathcal{A}_{\left[\mathsf{x}\mid\mathsf{\Theta}\right]\varepsilon}^{N}\left(\boldsymbol{\theta}^{N}\right)}^{\boldsymbol{\theta}^{N}\in\mathcal{A}_{\left[\mathsf{\Theta}\right]\varepsilon}^{N}\setminus\left\{ \mathbf{\boldsymbol{\varphi}}^{N}\right\} }}p\left(\mathbf{z}_{1}^{N}\mid\mathbf{\boldsymbol{\varphi}}^{N}\right),\label{eq:PA11}
\end{equation}
which can be further decomposed as $P_{\mathcal{A}_{1}}=P_{\mathcal{A}_{11}}+P_{\mathcal{A}_{12}}$,
with 
\begin{eqnarray}
P_{\mathcal{A}_{11}} & = & \sum_{\mathbf{\boldsymbol{\varphi}}^{N}\in\mathcal{A}_{\left[\mathsf{\Theta}\right]\varepsilon}^{N}}p\left(\mathbf{\boldsymbol{\varphi}}^{N}\right)\sum_{_{\,\mathbf{z}_{1}^{N}\in\mathcal{A}_{\left[\mathsf{x}\mid\mathsf{\Theta}\right]\varepsilon}^{N}\left(\boldsymbol{\theta}^{N}\right)}^{\boldsymbol{\theta}^{N}\in\mathcal{A}_{\left[\mathsf{\Theta}\right]\varepsilon}^{N}\setminus\left\{ \mathbf{\boldsymbol{\varphi}}^{N}\right\} }}p\left(\mathbf{z}_{1}^{N}\mid\mathbf{\boldsymbol{\varphi}}^{N}\right)\nonumber \\
 & \overset{(a)}{\leq} & \sum_{\mathbf{\boldsymbol{\varphi}}^{N}\in\mathcal{A}_{\left[\mathsf{\Theta}\right]\varepsilon}^{N}}p\left(\mathbf{\boldsymbol{\varphi}}^{N}\right)\sum_{\mathbf{z}_{1}^{N}\in\mathbb{F}_{Q}^{N}\setminus\mathcal{A}_{\left[\mathsf{x}\mid\mathsf{\Theta}\right]\varepsilon}^{N}\left(\mathbf{\boldsymbol{\varphi}}^{N}\right)}p\left(\mathbf{z}_{1}^{N}\mid\mathbf{\boldsymbol{\varphi}}^{N}\right)\nonumber \\
 & \leq & \sum_{\mathbf{\boldsymbol{\varphi}}^{N}\in\mathcal{A}_{\left[\mathsf{\Theta}\right]\varepsilon}^{N}}p\left(\mathbf{\boldsymbol{\varphi}}^{N}\right)\varepsilon\leq\varepsilon,\label{eq:PA12}
\end{eqnarray}
where $(a)$ comes from the fact that if (\ref{eq:overlapping}) is
satisfied, one has 
\begin{equation}
\bigcup_{\boldsymbol{\theta}^{N}\in\mathcal{A}_{\left[\mathsf{\Theta}\right]\varepsilon}^{N}\setminus\left\{ \mathbf{\boldsymbol{\varphi}}^{N}\right\} }\mathcal{A}_{\left[\mathsf{x}\mid\mathsf{\Theta}\right]\varepsilon}^{N}\left(\boldsymbol{\theta}^{N}\right)\subseteq\mathbb{F}_{Q}^{N}\setminus\mathcal{A}_{\left[\mathsf{x}\mid\mathsf{\Theta}\right]\varepsilon}^{N}\left(\mathbf{\boldsymbol{\varphi}}^{N}\right).\label{eq:PA13}
\end{equation}
On the other hand, 
\begin{eqnarray}
P_{\mathcal{A}_{12}} & = & \sum_{\mathbf{\boldsymbol{\varphi}}^{N}\in\mathbb{F}_{Q}^{N}\setminus\mathcal{A}_{\left[\mathsf{\Theta}\right]\varepsilon}^{N}}p\left(\mathbf{\boldsymbol{\varphi}}^{N}\right)\sum_{_{\mathbf{z}_{1}^{N}\in\mathcal{A}_{\left[\mathsf{x}\mid\mathsf{\Theta}\right]\varepsilon}^{N}\left(\boldsymbol{\theta}^{N}\right)}^{\quad\:\boldsymbol{\theta}^{N}\in\mathcal{A}_{\left[\mathsf{\Theta}\right]\varepsilon}^{N}}}p\left(\mathbf{z}_{1}^{N}\mid\mathbf{\boldsymbol{\varphi}}^{N}\right)\nonumber \\
 & \leq & \sum_{\mathbf{\boldsymbol{\varphi}}^{N}\in\mathbb{F}_{Q}^{N}\setminus\mathcal{A}_{\left[\mathsf{\Theta}\right]\varepsilon}^{N}}p\left(\mathbf{\boldsymbol{\varphi}}^{N}\right)\leq\varepsilon,\label{eq:PA14}
\end{eqnarray}
since for this part
\begin{equation}
\bigcup_{\boldsymbol{\theta}^{N}\in\mathcal{A}_{\left[\mathsf{\Theta}\right]\varepsilon}^{N}}\mathcal{A}_{\left[\mathsf{x}\mid\mathsf{\Theta}\right]\varepsilon}^{N}\left(\boldsymbol{\theta}^{N}\right)\subseteq\mathbb{F}_{Q}^{N}.\label{eq:PA15}
\end{equation}
From (\ref{eq:PA12}) and (\ref{eq:PA14}), Lemma~4
is proved.
\end{proof}
Now consider the term (\ref{eq:PA2}), 
\begin{eqnarray}
P_{\mathcal{A}_{2}} & = &
\sum_{d=1}^{N}\sum_{_{\mathbf{\boldsymbol{\varphi}}^{N}\in\mathbb{F}_{Q}^{N}\setminus\left\{
\boldsymbol{\theta}^{N}\right\}
}^{\;\;\;\;\boldsymbol{\theta}^{N}\in\mathcal{A}_{\left[\mathsf{\Theta}\right]\varepsilon}^{N}}}\sum_{_{\mathbf{z}_{1}^{N}\in\mathbb{F}_{Q}^{N}:\left\Vert
\mathbf{x}^{N}-\mathbf{z}_{1}^{N}\right\Vert _{0}=d}^{~ ~ ~\mathbf{x}^{N}\in\mathcal{A}_{\left[\mathsf{x}\mid\mathsf{\Theta}\right]\varepsilon}^{N}\left(\boldsymbol{\theta}^{N}\right)}}p\left(\mathbf{\boldsymbol{\varphi}}^{N},\mathbf{z}_{1}^{N}\right)\cdot f\left(d,0;\gamma,Q,M\right)\nonumber \\
 & \leq & \sum_{d=1}^{\left\lfloor \beta N\right\rfloor }\sum_{_{\mathbf{\boldsymbol{\varphi}}^{N}\in\mathbb{F}_{Q}^{N}\setminus\left\{ \boldsymbol{\theta}^{N}\right\} }^{\;\;\;\;\boldsymbol{\theta}^{N}\in\mathcal{A}_{\left[\mathsf{\Theta}\right]\varepsilon}^{N}}}\sum_{\mathbf{z}_{1}^{N}\in\mathbb{F}_{Q}^{N}}\sum_{\mathbf{x}^{N}\in\mathbb{F}_{Q}^{N}:\left\Vert \mathbf{x}^{N}-\mathbf{z}_{1}^{N}\right\Vert _{0}=d}p\left(\mathbf{\boldsymbol{\varphi}}^{N},\mathbf{z}_{1}^{N}\right)\cdot f\left(1,0;\gamma,Q,M\right)\nonumber \\
 &  & +\sum_{_{\mathbf{x}^{N}\in\mathcal{A}_{\left[\mathsf{x}\mid\mathsf{\Theta}\right]\varepsilon}^{N}\left(\boldsymbol{\theta}^{N}\right)}^{\quad\boldsymbol{\theta}^{N}\in\mathcal{A}_{\left[\mathsf{\Theta}\right]\varepsilon}^{N}}}\sum_{_{\;\:\quad\,\mathbf{z}_{1}^{N}\in\mathbb{F}_{Q}^{N}}^{\mathbf{\boldsymbol{\varphi}}^{N}\in\mathcal{A}_{\left[\mathsf{\Theta}\right]\varepsilon}^{N}\setminus\left\{ \boldsymbol{\theta}^{N}\right\} }}p\left(\mathbf{\boldsymbol{\varphi}}^{N},\mathbf{z}_{1}^{N}\right)\cdot f\left(\left\lceil \beta N\right\rceil ,0;\gamma,Q,M\right),\label{eq:PA21}
\end{eqnarray}
which is similar to (\ref{eq:PeNCUpp5}) in Section~\ref{sec:SufficientConditionNC-1}.
For $N$ sufficient large, the condition on $M/N$ to ensure $P_{\mathcal{A}_{2}}$
tends to zero as $N\rightarrow\infty$ is
\begin{equation}
\frac{M}{N}>\frac{\mathcal{H}\left(\mathsf{\Theta},\mathsf{x}\right)+\varepsilon}{\log Q-\xi},\label{eq:SuffcientNS ratio}
\end{equation}
for some $\xi\in\mathbb{R}^{+}$. Finally, we have Proposition~\ref{Prop:Sufficient Condition NS}
to conclude the sufficient condition for reliable recovery in the
NS case. 
\begin{prop}[Sufficient condition, NS case]
\label{Prop:Sufficient Condition NS}In the NS case, fix an arbitrary
small positive real number $\delta$,  there exists $\varepsilon\in\mathbb{R}^{+}$,
$\xi\in\mathbb{R}^{+}$, $N_{\delta}\in\mathbb{N}^{+}$ and $M_{\varepsilon}\in\mathbb{N}^{+}$
such that for any $N>N_{\delta}$ and $M>M_{\varepsilon}$, one has
$P_{\textrm{e}}<\delta$ under MAP decoding if (\ref{eq:overlapping})
and (\ref{eq:SuffcientNS ratio}) hold. One can make both $\varepsilon$
and $\xi$ arbitrary close to $0$ as $N\rightarrow\infty$. 
\end{prop}
Finally, the NCS case, accounting for both communication and sensing
noise, has to be considered.
\begin{prop}[Sufficient condition, NCS case]
\label{Prop:Sufficient Condition NCS}Considering both communication
noise and sensing noise, for $N$ and $M$ sufficient large and positive
$\varepsilon$, $\xi$ arbitrary small, the reliable recovery can
be ensured under MAP decoding if \end{prop}
\begin{itemize}
				\item \textit{the communication noise is not uniformly distributed,
								(\ref{eq:SufficientNoise})}
\item \textit{there is no overlapping between any two different weakly conditional
typical sets, i.e., $\mathcal{A}_{\left[\mathsf{x}\mid\mathsf{\Theta}\right]\varepsilon}^{N}\left(\boldsymbol{\theta}^{N}\right)\cap\mathcal{A}_{\left[\mathsf{x}\mid\mathsf{\Theta}\right]\varepsilon}^{N}\left(\mathbf{\boldsymbol{\varphi}}^{N}\right)=\emptyset$
for any two typical but different $\boldsymbol{\theta}^{N}$ and $\mathbf{\boldsymbol{\varphi}}^{N}$,} 
\item \noindent \textit{the sparsity factor satisfies the constraint in
(\ref{eq:SufficientGamma}),}
\item \noindent \textit{the compression ratio $M/N$ is lower bounded by
\begin{equation}
\frac{M}{N}>\frac{\mathcal{H}\left(\mathsf{\Theta},\mathsf{x}\right)+\varepsilon}{\log Q-H\left(p_{\mathsf{u}}\right)-\xi},\label{eq:Sufficient_ratio_NCS}
\end{equation}
}
\end{itemize}
The derivations are similar to those of Proposition~\ref{Prop:SufficientCond NC}
and Proposition~\ref{Prop:Sufficient Condition NS}.

\subsection{Discussion and Numerical Results}

When comparing the necessary condition in Proposition~\ref{Prop:CNforNCSCase}
and the sufficient condition in Proposition~\ref{Prop:Sufficient Condition NCS},
an interesting fact is that $\mathcal{H}\left(\mathsf{\Theta}\mid\mathsf{x}\right)=0$
is a sufficient condition to have (\ref{eq:overlapping}). This implies
that the value of $\boldsymbol{\theta}^{N}$ should be fixed almost
surely, as long as $\mathbf{x}^{N}$ is known. So, (\ref{eq:overlapping})
is helpful to interpret (\ref{eq:necessary1}), justifying the need
for the conditional entropy $\mathcal{H}\left(\mathsf{\Theta}\mid\mathsf{x}\right)$
to tend to zero as $N$ increases. This condition may be satisfied
since $\left|\mathcal{A}_{\left[\mathsf{\Theta}\right]\varepsilon}^{N}\right|\ll\left|\mathbb{F}_{Q}^{N}\right|$
as long as $\mathcal{H}\left(\mathsf{\Theta}\right)<\log Q$. The
entropy rate $\mathcal{H}\left(\mathsf{\Theta}\right)$ can be very
small, Appendix~\ref{sec:PossibleSituation} presents a possible
situation where $\mathcal{H}\left(\mathsf{\Theta}\right)=0$. Another
implicit constraint resulting from {(\ref{eq:overlapping})}
is 
\begin{equation}
\sum_{\boldsymbol{\theta}^{N}\in\mathcal{A}_{\left[\mathsf{\Theta}\right]\varepsilon}^{N}}\mathcal{A}_{\left[\mathsf{x}\mid\mathsf{\Theta}\right]\varepsilon}^{N}\left(\boldsymbol{\theta}^{N}\right)\leq\left|\mathbb{F}_{Q}^{N}\right|
\end{equation}
which means that 
\begin{equation}
\mathcal{H}\left(\mathsf{\Theta},\mathsf{x}\right)\leq\log Q.\label{eq:rate0-1}
\end{equation}
Consider a communication noise with $\Pr\left(\mathsf{u}\neq0\right)=0.1$
and the transition pmf
\begin{equation}
p(x_{n}\mid\theta_{n})=\begin{cases}
1-\Pr\left(\mathsf{x}\neq\mathsf{\Theta}\right) & \textrm{if }x_{n}=\theta_{n}\\
\frac{\Pr\left(\mathsf{x}\neq\mathsf{\Theta}\right)}{Q-1} & \textrm{if }x_{n}\in\mathbb{F}_{Q}^{N}\setminus\left\{ \theta_{n}\right\} 
\end{cases},
\end{equation}
where $\Pr\left(\mathsf{x}\neq\mathsf{\Theta}\right)$ denotes the
probability of the sensing error. In Figure~\ref{fig:Optimum-compression-ratio-NCS},
the lower bound of $M/N$ is represented as a function of $\mathcal{H}\left(\mathsf{\mathsf{\Theta}}\right)/\log Q$,
for different values of $Q$ and for different values of $\Pr\left(\mathsf{x}\neq\mathsf{\Theta}\right)$. 

\begin{figure}[th]
\centering
\includegraphics[width=1\columnwidth]{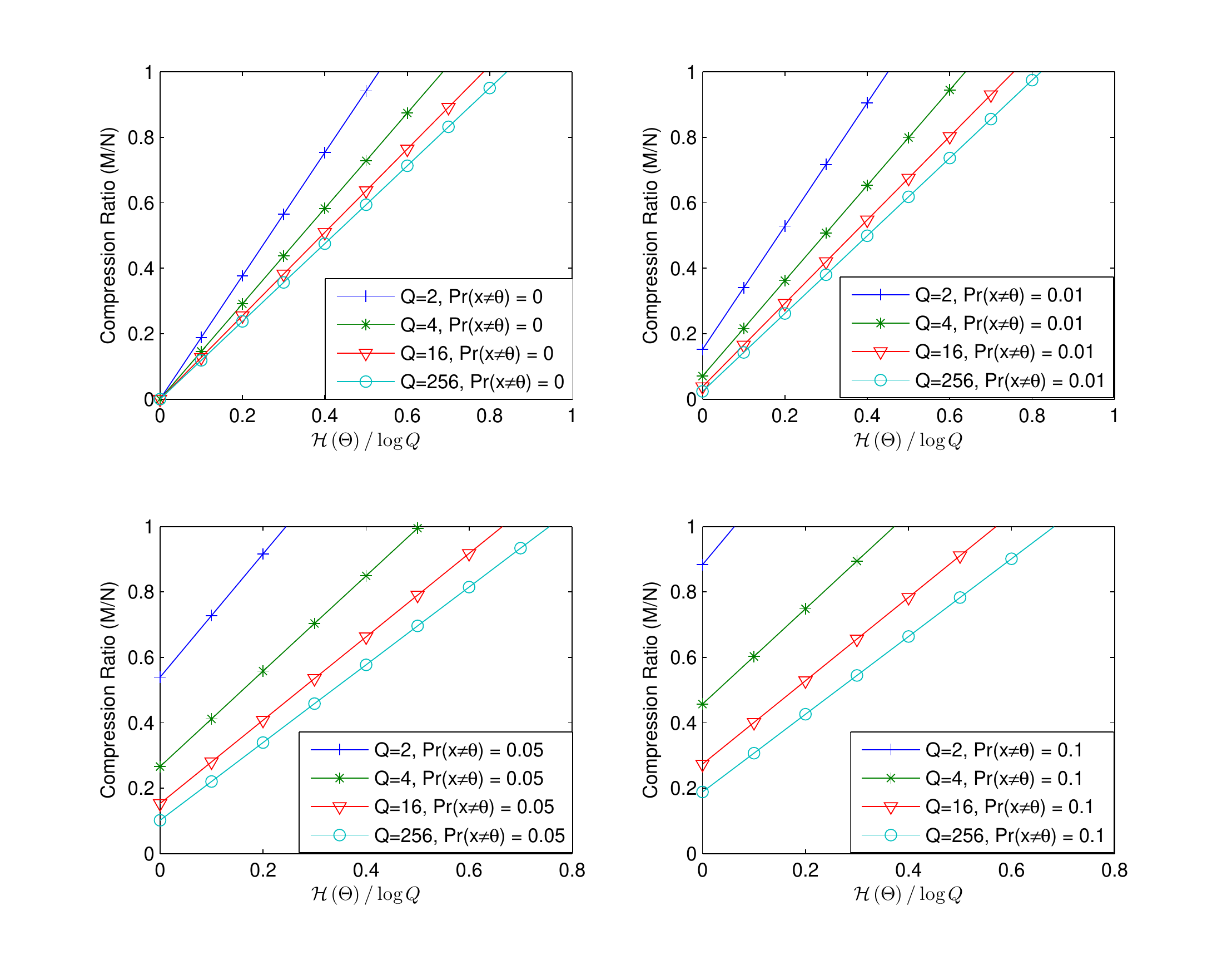}

\caption{Optimum achievable compression ratio in function of $\mathcal{H}\left(\mathsf{\mathsf{\Theta}}\right)/\log Q$,
according to (\ref{eq:Sufficient_ratio_NCS}), for the cases that
$\Pr\left(\mathsf{x}\neq\mathsf{\Theta}\right)$ being $0$ (NC case),
$0.01$, $0.05$, and $0.1$, respectively, when $\Pr\left(\mathsf{u}\neq0\right)=0.1$
\label{fig:Optimum-compression-ratio-NCS}}
\end{figure}

\section{Conclusions and future work}

\label{sec:Conclusions}

{In this paper we have considered robust Bayesian compressed
sensing over finite fields under MAP decoding. Both asymptotically
necessary and sufficient conditions of the compression ratio for reliable
recovery are obtained and their convergence is also shown, even in
the case of sparse sensing matrices. Several previous results have
been generalized by considering a stationary and ergodic source model.
Both communication noise and sensing noise have been taken into account.
We have shown that the choice of the sparsity factor of the sensing
matrix only depends on the communication noise. Since necessary and
sufficient conditions asymptotically converge, the MAP decoder achieves
the optimum lower bound of the compression ratio, which can be expressed
as a function of $\mathcal{H}\left(\mathsf{\Theta},\mathsf{x}\right)$,
$H\left(p_{\mathsf{u}}\right)$, and the alphabet size.}

{In this paper, the sensing matrix was assumed to be
perfectly known, without specific structure. In sensor network compressive
sensing applications, the structure of the sensing matrix usually
depends on the structure of the network. Evaluating the impact of
these constraints on the compression efficiency will be the subject
of future research. A first step in this direction was done in \cite{giannakis},
which considered clustered sensors.} 

\appendices

\section{Proof of Lemma 1}

\label{sec:Proof_lemma1}
\begin{proof}
Let $\boldsymbol{\mathsf{A}}_{i}$ be the $i$-th row of $\boldsymbol{\mathsf{A}}$.
As all entries in $\boldsymbol{\mathsf{A}}$ are independent 
\begin{equation}
\Pr\left\{ \boldsymbol{\mathsf{A}}\mathbf{\boldsymbol{\mu}}^{N}=\mathbf{s}^{M}\mid\mathbf{\boldsymbol{\mu}}^{N}\neq\mathbf{0},\mathbf{s}^{M}\right\} =\prod_{i=1}^{M}\Pr\left\{ \boldsymbol{\mathsf{A}}_{i}\mathbf{\boldsymbol{\mu}}^{N}=s_{i}\mid\mathbf{\boldsymbol{\mu}}^{N}\neq\mathbf{0},s_{i}\right\} .\label{eq:ind}
\end{equation}
According to \cite[Lemma 21]{rankmin}, we have 
\begin{equation}
\Pr\left\{ \boldsymbol{\mathsf{A}}_{i}\mathbf{\boldsymbol{\mu}}^{N}=0\mid\left\Vert \mathbf{\boldsymbol{\mu}}^{N}\right\Vert _{0}=d_{1}\right\} =Q^{-1}+\left(1-\frac{\gamma}{1-Q^{-1}}\right)^{d_{1}}\left(1-Q^{-1}\right),\label{eq:p30}
\end{equation}
and 
\begin{equation}
\Pr\left\{ \boldsymbol{\mathsf{A}}_{i}\mathbf{\boldsymbol{\mu}}^{N}=q\mid\left\Vert \mathbf{\boldsymbol{\mu}}^{N}\right\Vert _{0}=d_{1},q\in\mathbb{F}_{Q}\setminus\left\{ 0\right\} \right\} =Q^{-1}-\left(1-\frac{\gamma}{1-Q^{-1}}\right)^{d_{1}}Q^{-1}.\label{eq:p30-1}
\end{equation}
Since $d_{2}$ is the number of non-zero entries of $\mathbf{s}^{M}$,
combining (\ref{eq:ind}), (\ref{eq:p30}), and (\ref{eq:p30-1}),
one gets\textit{
\begin{eqnarray}
f\left(d_{1},d_{2};\gamma,Q,M\right) &=& \nonumber\\
&&\hspace{-4cm}\left(Q^{-1}+\left(1-\frac{\gamma}{1-Q^{-1}}\right)^{d_{1}}\left(1-Q^{-1}\right)\right)^{M-d_{2}}\left(Q^{-1}-\left(1-\frac{\gamma}{1-Q^{-1}}\right)^{d_{1}}Q^{-1}\right)^{d_{2}}.\label{eq:f122}
\end{eqnarray}
}The monotonicity of this function is not hard to obtain with its
expression and the condition (\ref{eq:gama}).
\end{proof}

\section{A Possible Situation for $\mathcal{H}\left(\mathsf{\Theta}\right)=0$}

\label{sec:PossibleSituation}

Consider $N$ sensors uniformly deployed over a unit-radius disk.
The physical quantities (in $\mathbb{R}$), which are collected by
the sensors, are denoted by $\boldsymbol{\mathsf{\Omega}}^{N}\in\mathbb{R}^{N}$.
We assume that $\boldsymbol{\mathsf{\Omega}}^{N}\sim\mathcal{N}\left(0,\boldsymbol{\Sigma}\right)$
with 
\begin{equation}
\boldsymbol{\Sigma}=\left[\begin{array}{cccc}
1 & \textrm{e}^{-\lambda d_{1,2}^{2}} & \cdots & \textrm{e}^{-\lambda d_{1,N}^{2}}\\
\textrm{e}^{-\lambda d_{2,1}^{2}} & 1 &  & \textrm{e}^{-\lambda d_{2,N}^{2}}\\
\vdots &  & \ddots & \vdots\\
\textrm{e}^{-\lambda d_{N,1}^{2}} & \cdots & \cdots & 1
\end{array}\right],
\end{equation}
where $\lambda$ is some constant, $d_{i,j}$ is the distance between
sensors $i$ and $j$. The distance between two sensors is random
since the location of each sensor is random. The real-valued entries
of $\boldsymbol{\mathsf{\Omega}}^{N}$ are quantized with a $Q-$level
scalar quantizer. We assume that $Q=2$, corresponding to the rule
\begin{equation}
\mathsf{\Theta}_{i}=\begin{cases}
0 & \textrm{if }\mathsf{\Omega}_{i}<0,\\
1 & \textrm{if }\mathsf{\Omega}_{i}\geq0.
\end{cases}\label{eq:rule}
\end{equation}
With the above assumptions, we can prove the following lemma.

\bigskip{}

\noindent \textbf{Lemma 5}. \textit{The conditional entropy $H\left(\mathsf{\Theta}_{n}\mid\boldsymbol{\mathsf{\Theta}}_{1}^{n-1}\right)$
converges to zero for $n\rightarrow\infty$ .}
\begin{proof}
Suppose that $j$ is the index of the sensor which has the minimum
distance to sensor $n$, among the $n-1$ neighbor sensors, $i.e.$,
\begin{equation}
j=\textrm{arg}\min_{1\leq i\leq n-1}d_{n,i}.
\end{equation}
We have 
\begin{equation}
H\left(\mathsf{\Theta}_{n}\mid\boldsymbol{\mathsf{\Theta}}_{1}^{n-1}\right)\leq H\left(\mathsf{\Theta}_{n}\mid\mathsf{\Theta}_{j}\right)\label{eq:thetaij}
\end{equation}
Denote the minimum distance as $\underline{d}\left(n\right)=d_{n,j}$,
the covariance matrix of $\mathsf{\Omega}_{n}$ and $\mathsf{\Omega}_{j}$
is 
\begin{equation}
\boldsymbol{\Sigma}_{n}=\left[\begin{array}{cc}
1 & \rho\\
\rho & 1
\end{array}\right],
\end{equation}
where $\rho=\textrm{e}^{-\lambda\underline{d}\left(n\right)^{2}}$.
For a pair of realizations $\omega_{n}$ and $\omega_{j}$, the joint
probability density function writes
\begin{equation}
g\left(\omega_{n},\omega_{j}\right)=\frac{1}{2\pi\sqrt{1-\rho^{2}}}\exp\left(-\frac{\omega_{n}^{2}+\omega_{j}^{2}-2\rho\omega_{n}\omega_{j}}{2\left(1-\rho^{2}\right)}\right).
\end{equation}
We easily obtain the probability of both $\mathsf{\Omega}_{n}$ and
$\mathsf{\Omega}_{j}$ being negative,
\begin{eqnarray}
\Pr\left\{ \mathsf{\Omega}_{n}<0\textrm{ and }\mathsf{\Omega}_{j}<0\right\}  & = & \int_{-\infty}^{0}\int_{-\infty}^{0}g\left(\omega_{n},\omega_{j}\right)\textrm{d}\omega_{n}\textrm{d}\omega_{j}\nonumber \\
 & = & \frac{1}{4}+\frac{1}{2\pi}\arctan\left(\frac{\rho}{\sqrt{1-\rho^{2}}}\right)\label{eq:00}
\end{eqnarray}
Taking into account (\ref{eq:rule}), (\ref{eq:00}) is exactly the
probability of the pair $\left(\mathsf{\Theta}_{n},\mathsf{\Theta}_{j}\right)$
being $\left(0,0\right)$. Define 
\begin{equation}
\varepsilon\left(\rho\right):=\frac{1}{4}-\frac{1}{2\pi}\arctan\left(\frac{\rho}{\sqrt{1-\rho^{2}}}\right).
\end{equation}
After the similar derivations, one obtains 
\begin{equation}
\Pr\left(\mathsf{\Theta}_{n}=0,\mathsf{\Theta}_{j}=0\right)=\Pr\left(\mathsf{\Theta}_{n}=1,\mathsf{\Theta}_{j}=1\right)=\frac{1}{2}-\varepsilon\left(\rho\right)
\end{equation}
 and 
\begin{equation}
\Pr\left(\mathsf{\Theta}_{n}=0,\mathsf{\Theta}_{j}=1\right)=\Pr\left(\mathsf{\Theta}_{n}=1,\mathsf{\Theta}_{j}=0\right)=\varepsilon\left(\rho\right).
\end{equation}
Then the joint entropy is 
\begin{equation}
H\left(\mathsf{\Theta}_{n},\mathsf{\Theta}_{j}\right)=1+H_{2}\left(2\varepsilon\left(\rho\right)\right).
\end{equation}
Meanwhile $H\left(\mathsf{\Theta}_{j}\right)=1$, thanks to the 2-level
uniform quantizer. Obviously 
\begin{equation}
H\left(\mathsf{\Theta}_{n}\mid\mathsf{\Theta}_{j}\right)=H_{2}\left(2\varepsilon\left(\rho\right)\right)=H_{2}\left(2\varepsilon\left(\textrm{e}^{-\lambda\underline{d}^{2}}\right)\right).\label{eq:cond_ent}
\end{equation}
This conditional entropy is increasing in $\underline{d}$. When the
number of sensors increases, the disk will be denser, and the minimum
distance $\underline{d}$ goes smaller. Thus, $\underline{d}$ tends
to 0 as $n\rightarrow\infty$, which implies that $H\left(\mathsf{\Theta}_{n}\mid\mathsf{\Theta}_{j}\right)\rightarrow0$.
According to (\ref{eq:thetaij}), we conclude that $H\left(\mathsf{\Theta}_{n}\mid\boldsymbol{\mathsf{\Theta}}_{1}^{n-1}\right)$
also goes to zero as $n\rightarrow\infty$.
\end{proof}
Applying the chain rule, the entropy rate writes
\begin{equation}
\mathcal{H}\left(\mathsf{\Theta}\right)=\lim_{N\rightarrow\infty}\frac{H\left(\mathsf{\Theta}_{1}\right)+\sum_{n=2}^{N}H\left(\mathsf{\Theta}_{n}\mid\boldsymbol{\mathsf{\Theta}}_{1}^{n-1}\right)}{N}.
\end{equation}
By Cesaro mean \cite[Theorem 4.2.3]{EIT}, $\mathcal{H}\left(\mathsf{\Theta}\right)=0$
as $H\left(\mathsf{\Theta}_{n}\mid\boldsymbol{\mathsf{\Theta}}_{1}^{n-1}\right)\rightarrow0$. 

\bibliographystyle{IEEETran}
\bibliography{BiblioWenjie}

\end{document}